\documentclass[11pt, onecolumn]{article}
\usepackage[top=1in, bottom=1in, left=1.25in, right=1.25in]{geometry}

\usepackage{amsfonts}
\usepackage[cmex10]{amsmath}
\usepackage{amssymb}
\usepackage{amsthm}
\usepackage{cite}
\usepackage{color,soul}
\usepackage{graphicx}
\usepackage{bm}
\usepackage{booktabs}
\usepackage{algorithm}
\usepackage{algorithmic}
\usepackage{mdwmath}
\usepackage{mdwtab}
\usepackage{multirow}
\usepackage{diagbox}
\usepackage{url}
\usepackage{enumitem}

\usepackage{float}
\usepackage{subcaption}
\usepackage{wrapfig}
\usepackage[export]{adjustbox}

\graphicspath {{figures/}}

\theoremstyle{definition}
\newtheorem{definition}{Definition}

\theoremstyle{definition}
\newtheorem{proposition}{Proposition}

\theoremstyle{definition}
\newtheorem{lemma}{Lemma}

\theoremstyle{definition}
\newtheorem{theorem}{Theorem}

\theoremstyle{definition}
\newtheorem{corollary}{Corollary}

\theoremstyle{remark}
\newtheorem{remark}{Remark}

\newcommand{\aff}{{\operatorname{aff}}}
\newcommand{\set}{\mathcal}
\newcommand{\probP}{{\mathbb P}}
\newcommand{\projP}{{\mathcal P}}
\newcommand{\affX}{\aff_{\set X}}

\newcommand{\affY}{\aff_{\set Y}}
\newcommand{\affYS}{\aff_{\set Y}^2}
\newcommand{\mphi}{\bm\Phi}
\newcommand{\ptwo}{\projP_{\set Y_2}}
\newcommand{\ptwop}{\projP_{\set Y_2^\perp}}
\newcommand{\pz}[1]{\projP_{\set Z_{#1}}}

\newcommand{\vect}[1]{{\mathbf{#1}}}
\newcommand{\x}{{\vect x}}
\newcommand{\y}{{\vect y}}
\newcommand{\q}{{\vect q}}
\renewcommand{\u}{{\vect u}}
\renewcommand{\v}{{\vect v}}
\newcommand{\A}{{\vect A}}
\newcommand{\D}{{\vect D}}

\newcommand{\complexityO}{{\mathcal{O}}}
\newcommand{\Q}{{\vect Q}}
\newcommand{\U}{{\vect U}}
\newcommand{\V}{{\vect V}}
\newcommand{\Z}{{\vect Z}}

\begin{document}
\title{Unraveling the Veil of Subspace RIP Through Near-Isometry on Subspaces}

\author{Xingyu~Xu, Gen~Li, and~Yuantao~Gu%
\thanks{ 
The authors are with Department of Electronic Engineering, Tsinghua University, Beijing 100084, China. 
The corresponding author of this paper is Y. Gu (gyt@tsinghua.edu.cn).}% <-this % stops a space
}
\date{Manuscript received April 24, 2019, revised September 22, 2019}

\maketitle

\begin{abstract}
Dimensionality reduction is a popular approach 
to tackle high-dimensional data with low-dimensional nature. 
Subspace Restricted Isometry Property, a newly-proposed concept,
has proved to be a useful tool in analyzing the effect 
of dimensionality reduction algorithms on subspaces. 
In this paper, 
we provide a characterization of subspace Restricted Isometry Property, 
asserting that matrices which act as a near-isometry on low-dimensional subspaces 
possess subspace Restricted Isometry Property. 
This points out a unified approach to discuss subspace Restricted Isometry Property. 
Its power is further demonstrated by the possibility to prove with this result 
the subspace RIP for a large variety of random matrices encountered in theory and practice, 
including subgaussian matrices, partial Fourier matrices, 
partial Hadamard matrices, partial circulant/Toeplitz matrices, 
matrices with independent strongly regular rows 
(for instance, matrices with independent entries having uniformly bounded $4+\epsilon$ moments), 
and log-concave ensembles. 
Thus our result could extend the applicability of random projections 
in subspace-based machine learning algorithms including subspace clustering
and allow for the application of some useful random matrices
which are easier to implement on hardware or are more efficient to compute. 

\end{abstract}

\section{Introduction} 
In this paper we investigate the subspace Restricted Isometry Property (RIP) of random projections 
and try to capture the root of subspace RIP.
In more intuitive language, given two linear subspaces in an ambient space, 
we ask that for which type of random projections, 
the ``distance'' of these two subspaces, when defined properly, 
is almost invariant after being projected. 
The precise meaning of these terms will be presented later in this section.
Before that, we ground our results with some preliminaries.
\subsection{Background}
High-dimensional signals can be computationally expensive, or even intractable to analyze. 
Fortunately, 
many real-world high-dimensional signals are of low-dimensional nature. 
In this vein, numerous low dimensional models have been proposed 
and have remarkably fascinated researches in signal processing \cite{Bruckstein2009Sparse,Baraniuk2009Random,Elad2010Sparse}.
Union of Subspaces (UoS) is a powerful low dimensional model which subsumes 
many classical models including sparse representation
and has been used extensively in the recent decade \cite{Eldar2009Robust}.  
Briefly speaking, 
UoS model assumes that in a dataset with high ambient dimension, 
the data points actually lie on a few low dimensional linear subspaces, 
and these subspaces characterize the intrinsic structure of the dataset. 

Subspace clustering \cite{Soltanolkotabi2012Geometric,Elhamifar2013Sparse,Soltanolkotabi2014Robust,Heckel2015Rsobust} 
is one of the various successful applications of the UoS model
that has achieved impressive performance in tasks such as motion segmentation,
face clustering, and anomaly detection. 
Moreover, the performance of subspace clustering 
is theoretically guaranteed under fairly general conditions, 
a fact proved in \cite{Soltanolkotabi2012Geometric} 
based on the concept of affinity, c.f. Definition \ref{def:affinity}. 
However, for traditional subspace clustering algorithms
there is a high computational cost in building the so-called similarity representation 
when the dataset is of high dimension. 
This defect can be overcome by random compression, 
as was done in Compressed Subspace Clustering (CSC) \cite{Mao2014Compressed, Meng2018CSC}. 
While random compression can significantly reduce the computational burden, 
it raised a new concern that the affinity between two subspaces 
may not be preserved after random compression, 
hence it is not clear whether there is a theoretical guarantee for CSC. 

Part of the above concern was addressed in \cite{Heckel2014Subspace, Heckel2017Dimensionality, Wang2019Theoretical}, 
which provided theoretical analyses 
for several popular CSC algorithms. 
However, these analyses are done per algorithm 
and do not focus on the concept of affinity. 
A theorem on ``invariance property'' of affinity under random projections 
would constitute a more universal framework to analyze the performance of CSC. 
Such a theorem was given in \cite{Li2018Restricted, Li2019Rigorous}, 
which basically states that the change of affinity between two subspaces 
is small with high probability under Gaussian random projections. 
Since affinity is closely related to the notion of projection Frobenius-norm distance between subspaces, 
this implies that the projection Frobenius-norm distance between subspaces is 
approximately preserved by a Gaussian random projection, 
a property named by \emph{subspace Restricted Isometry Property} (subspace RIP) 
resembling the classical RIP for sparse vectors.

This paper is devoted to a thorough investigation of subspace RIP. 
Our first aim is to answer the question: 
what should be the proper abstract setting to study subspace RIP, 
or more precisely,
what is the essential property of a matrix which would lead to subspace RIP? 
We will prove that such essential property is that 
the matrix acts as a near-isometry on any low-dimensional subspace. 
This is not obvious and requires involved analysis.
In fact, a naive argument using near-isometry will lose a factor 
of the dimension of the subspaces, hence will be far from optimal.
This fundamental result will be used to prove the subspace RIP for 
a large variety of random matrices, 
including subgaussian matrices and other random matrices with exponential Johnson-Lindenstrauss property, 
partial Fourier/Hadamard matrices and other randomly sampled Bounded Orthonormal Systems (BOS) \cite{Foucart2017Mathematical}, 
partial circulant/Toeplitz matrices \cite{Rauhut2012Restricted}, 
and also some typical heavy-tailed matrices, 
e.g. those with independent strongly regular rows \cite{Srivastava2013Covariance} 
or log-concave ensembles \cite{Adamczak2010Quantitative, Adamczak2011Sharp}. 
These results provide a universal framework to analyze the subspace RIP of random matrices 
and their effects on subspace related tasks, 
which requires rather weak assumptions on the random matrix 
but yields universal performance guarantee that are not constrained to specific algorithms.

{\color{black}
\subsection{Our Contribution}
In this paper we proved that the essential property of a matrix that leads to subspace RIP 
is that the matrix acts as a near-isometry on any low-dimensional subspace. 
This accounts for the root of subspace RIP and 
provides the proper abstract setting, or a unified approach, to discuss subspace RIP. 
Both the statement and the proof of this result are deterministic, 
thus apparently differ from previous work on subspace RIP \cite{Li2018Restricted, Li2019Rigorous}
which relied heavily on delicate probabilistic analysis of Gaussian matrices 
and cannot be decoupled into deterministic and probabilistic parts in an obvious way; 
it is even not clear how the proof there generalizes to subgaussian matrices. 
More discussions on this difference are carried out after sufficient technical preparation, 
in Section \ref{sec:discussions}.

With this result, 
we are able to provide an easy proof of subspace RIP 
for random matrices with exponential Johnson-Lindenstrauss property, e.g. subgaussian matrices, 
which generalizes the conclusion of \cite{Li2018Restricted, Li2019Rigorous}. 
Moreover, we will also prove that randomly sampled BOS, e.g. partial Fourier/Hadamard matrices, 
and partial circulant/Toeplitz matrices possess subspace RIP. 
These are matrices with fast matrix-vector multiplication algorithms that permit a wide application 
and could significantly accelerate computation in practice, 
and our results validate their use in subspace related tasks. 
Note that in \cite{Heckel2017Dimensionality} 
it was claimed that randomly sampled BOS could be used for random compression in subspace clustering 
meanwhile keeping the clustering performance,
but the proof was based on the assertion that randomly sampled BOS 
satisfies exponential Johnson-Lindenstrauss property, 
which was not provided with a legitimate proof there. 
The proof strategy in \cite{Heckel2017Dimensionality} appears only feasible 
to show that randomly sampled BOS satisfies exponential Johnson-Lindenstrauss property 
with unreasonably small constants, 
which is not helpful in practice. 
As such, our results constitute a more effective guarantee for 
performance of partial Fourier/Hadamard matrices and partial circulant/Toeplitz matrices 
on subspace clustering.

Recently, there are rising interests on heavy-tailed random matrices. 
We will deal with two typical types of such random matrices, 
namely the ones with finite $4+\epsilon$ moments and log-concave ensembles, 
and show that how a combination of our characterization of subspace RIP 
and well-known results in covariance estimation implies easily 
the subspace RIP of these random matrices. 

From a practical point of view, 
our result holds for much more general random matrices 
compared with the subspace RIP for Gaussian random matrices in \cite{Li2019Rigorous},
hence allows the application of random matrices that are more useful in practice, 
for instance those matrices which are easier to generate and store on hardware, 
e.g. Bernoulli matrices, 
or those who arise in the physical world naturally and are more efficient to compute, 
e.g. partial Fourier/Hadamard matrices and partial circulant/Toeplitz matrices. 
Most of these matrices are inaccessible within the proof strategy in 
previous works on subspace RIP \cite{Li2018Restricted, Li2019Rigorous}.
As pointed out in \cite{Hinojosa2018Coded}, 
in applications such as compressive spectral imaging, 
typical projection matrices are not Gaussian. 
Instead, Bernoulli matrices can be used \cite{Martin2016Hyperspectral}. 
Our result demonstrate more practical scenarios where techniques of random projections
and in particular, CSC algorithms may apply. 
}

\subsection{Notations and Conventions}
Throughout this paper, $c$ and $C$ denote two positive universal constants 
that may vary upon each appearance, 
while $\tilde c$ is the constant appearing in the definition 
of exponential Johnson-Lindenstrauss property, c.f. Definition \ref{def:jl-property}.
Bold upper case letters, e.g. $\A$, are used to denote a matrix, 
while bold lower case letters, e.g. $\u$, are used to denote a vector. 
$\mphi$ will always be a random matrix. 
If $\set X$ is a linear subspace of $\mathbb R^N$, 
$\set X^\perp$ denotes its orthogonal complement. 
Orthogonal projections onto subspace $\set X$ will be denoted by $\mathcal P_{\set X}$.
The maximal and minimal singular value of a matrix $\A$ will be denoted by 
$s_{\max}(\A)$ and $s_{\min}(\A)$. 
$\|\v\|$ is the Euclidean norm of the vector $\v$, 
and $\|\A\|_{\rm F}$ is the Frobenius norm of the matrix $\A$. 
The $(n-1)$-dimensional unit sphere in $\mathbb R^n$ is denoted by $\mathbb S^{n-1}$, 
i.e. $\mathbb S^{n-1}=\{\x\in\mathbb R^n:\|\x\|=1\}$.
The affinity between subspaces $\set X_1$, $\set X_2$, 
defined in Definition \ref{def:affinity}, 
will be denoted by $\aff(\set X_1,\set X_2)$.
Occasionally we will write $\affX$ (resp. $\affY$) as an abbreviation of 
$\aff(\set X_1, \set X_2)$ (resp. $\aff(\set Y_1,\set Y_2)$). 
The probability of an event is denoted by $\probP(\cdot)$.
The expectation of a random variable/vector/matrix is denoted by $\mathbb E(\cdot)$. 

We will be a bit blurry when using ``infinitesimal'' $\varepsilon$. 
That is, 
we will implicitly shrink the value of $\varepsilon$ by a constant ratio when needed. 
For example, we will assert $\probP(X>\varepsilon)<{\rm e}^{-c\varepsilon^2n}$ 
while we actually proved $\probP(X>2\varepsilon)<{\rm e}^{-c\varepsilon^2n}$. 
Such gaps are usually easy to fill and harmless to skip.
In fact, 
the former statement can be easily derived 
from the latter by replacing $\varepsilon$ with $\varepsilon/2$ 
and replacing $c$ with $4c$.

\subsection{Organization}
{\color{black}
The rest of this paper is organized as follows. 
In Section \ref{sec:preliminaries}, 
definitions and basic properties of affinity are provided. 
In Section \ref{sec:main-results}, 
we state our main theorem that a matrix acting as a near-isometry on a subspace $\set X$
preserves the affinity and projection Frobenius-norm distance between any pair of subspaces in $\set X$. 
Using this theorem, we analyze several important classes of random matrices in Section \ref{sec:examples} 
and prove their subspace RIP.
Section \ref{sec:proof} is devoted to the proof 
of the main theorem. 
Section \ref{sec:discussions} provides some further 
comments on proof strategies 
and comparison with related works. 
Section \ref{sec:applications} briefly introduces some examples among the various potential applications 
of our theory.
Section \ref{sec:simulations} verifies our results on a real-world dataset. 
Finally in Section \ref{sec:conclusion} we conclude the paper.
}
% -----------------------------------------------------------

\section{Preliminaries}\label{sec:preliminaries}

A key ingredient in the statement of our results 
is the \emph{affinity} between two subspaces, 
defined as following 
\cite{Soltanolkotabi2012Geometric,Heckel2017Dimensionality, Li2019Rigorous}:
\begin{definition}\label{def:affinity}
Let $\set X_1$, $\set X_2$ be subspaces of dimension $d_1$, $d_2$ in $\mathbb R^n$. 
Denote by $\projP_{\set X_1}$, $\projP_{\set X_2}$ the matrix of orthogonal projection 
onto $\set X_1$ and $\set X_2$. 
The affinity between $\set X_1$ and $\set X_2$ is
\begin{equation*}
    \aff(\set X_1,\set X_2)=\sqrt{\operatorname{tr}(\projP_{\set X_1}\projP_{\set X_2})}.
\end{equation*}
\end{definition}

There are several alternative ways to compute the affinity 
which will be used interchangeably. 
They are summarized in the following lemma. 
\begin{lemma}\label{lem:affinity-compute}
Let $\set X_1$, $\set X_2$ be subspaces of dimension $d_1$, $d_2$ in $\mathbb R^n$. 
Denote by $\projP_{\set X_1}$, $\projP_{\set X_2}$ the orthogonal projection 
onto $\set X_1$ and $\set X_2$.
\begin{enumerate}[label=\roman*)]
\item If $\U_1$, $\U_2$ are orthonormal bases of $\set X_1$, $\set X_2$, 
then 
\begin{equation*}
  \aff(\set X_1,\set X_2)=\|\U_1^{\rm T}\U_2\|_{\rm F},
\end{equation*}
where $\|\cdot\|_{\rm F}$ is the Frobenius norm. 
\item If $\U_2$ are orthonormal bases of $\set X_2$, 
then 
\begin{equation*}
  \aff(\set X_1,\set X_2)=\|\projP_{\set X_1}\U_2\|_{\rm F}.
\end{equation*}
\item There exists orthonormal bases $\U_1$, $\U_2$ of $\set X_1$, $\set X_2$ 
and nonnegative real numbers $\lambda_1\ge\lambda_2\ge\ldots\ge\lambda_{\min(d_1,d_2)}$, 
such that
\begin{equation*}
  \langle\u_{1,i},\u_{2,j}\rangle=\begin{cases}
    \lambda_i,\quad &i=j;\\
    0,\quad &i\ne j,
  \end{cases}
\end{equation*}
where $\u_{1,i}$,$\u_{2,j}$ denotes the $i$-th column of $\U_1$ 
and the $j$-th column of $\U_2$ respectively. 
Such $\U_1$, $\U_2$ are called \emph{principal orthonormal bases} 
of $\set X_1$, $\set X_2$.
Furthermore, 
\begin{equation*}
    \aff^2(\set X_1,\set X_2)=\sum_{i=1}^{\min(d_1,d_2)}\lambda_i^2.
\end{equation*}
\end{enumerate}
\end{lemma}

As its name suggests, affinity measures how close two subspaces are to each other. 
A relevant notion is the \emph{projection Frobenius-norm distance} of two subspaces \cite{Li2019Rigorous}. 
\begin{definition}\label{def:frob-dist}
The projection Frobenius-norm distance of two subspaces $\set X_1$,$\set X_2$ is defined as 
\[
    D(\set X_1,\set X_2)=\frac{1}{\sqrt2}\|\projP_{\set X_1}-\projP_{\set X_2}\|_{\rm F},
\]
where $\projP_{\set X_i}$ is the matrix of orthogonal projection onto $\set X_i$, $i=1,2$.
\end{definition}

Affinity and projection Frobenius-norm distance are related by 
\begin{equation}\label{eqn:affinity-and-frobenius-dist}
    D^2(\set X_1,\set X_2)=\frac{d_1+d_2}2-\aff^2(\set X_1,\set X_2).
\end{equation}

Intuitively, this means that the closer (in affinity) two subspaces are to each other, 
the less distant (in projection Frobenius-norm) they are to each other, which sounds tautological. 
Our main results will be stated based on affinity, 
but they can be easily translated to statements on projection Frobenius-norm distance
by (\ref{eqn:affinity-and-frobenius-dist}). 

We are now in the position to state the main result of \cite{Li2019Rigorous}.
\begin{theorem}\label{thm:li}
Assume $\mphi$ is an $n\times N$ Gaussian matrix with i.i.d. entries sampled from $\mathcal N(0,1/n)$. 
For any two subspaces $\set X_1$, $\set X_2$ of dimension $d_1$, $d_2$ in $\mathbb R^N$, 
assuming $d_1\le d_2$, 
denote by $\set Y_1$, $\set Y_2$ the image of $\set X_1$, $\set X_2$ under $\mphi$. 
Then for any $0<\varepsilon<1/2$ there exists 
positive constants $c_1(\varepsilon)$, $c_2(\varepsilon)$,  
such that for $n>c_1(\varepsilon)d_2$, 
the following is true with probability exceeding $1-{\rm e}^{c_2(\varepsilon)n}$.
\begin{equation}\label{eqn:li-affinity}
    \left|\aff^2(\set Y_1,\set Y_2)-\aff^2(\set X_1,\set X_2)\right|\le\left(d_1-\aff^2(\set X_1,\set X_2)\right)\varepsilon.
\end{equation}
\end{theorem}
\begin{remark}
Using the notion of projection Frobenius-norm distance 
(\ref{eqn:li-affinity}) has the following corollary in an easy-to-remember form:
\begin{equation}\label{eqn:li-frob-dist}
    \left|D^2(\set Y_1,\set Y_2)-D^2(\set X_1,\set X_2)\right|\le\varepsilon D^2(\set X_1,\set X_2).
\end{equation}
In other words, the distance of two subspaces only changes by a small portion 
after random projections with overwhelming probability. 
We thus call the ``affinity preserving'' property in (\ref{eqn:li-affinity}) 
by \emph{subspace Restricted Isometry property} (subspace RIP), 
a term resembling the classical Restricted Isometry Property for sparse vectors \cite{Candes2008Restricted}.
\end{remark}

The aim of this paper is to illuminate the root 
of subspace RIP and to extend Theorem \ref{thm:li} 
to a much wider range of random matrices that are more useful in practice.

% ---------------------------------------------------------

\section{Main Theorem}\label{sec:main-results}
{\color{black}
Lying in the center of our theory is the following theorem:
\begin{theorem}\label{thm:cov-est-to-subspace-rip}
Let $\set X$ be a $d$-dimensional subspace in $\mathbb R^N$.
Let $\set X_1$, $\set X_2$ be subspaces in $\set X$ whose dimensions are respectively 
$d_1$ and $d_2$,
and (without loss of generality) assume that $d_1\le d_2$. 
Denote by $\U$ a matrix whose columns constitute an orthonormal basis of $\set X$.
Suppose $\mphi$ is a $n\times N$ matrix satisfying for some $\delta\in(0,1/4)$ that
\begin{equation}\label{eqn:orthonormal-perturbation-bound}
  1-\delta<s^2_{\min}(\mphi\U)\le s^2_{\max}(\mphi\U)<1+\delta.
\end{equation}
Then with $\set Y_i=\mphi\set X_i$, 
$\affY=\aff(\set Y_1,\set Y_2)$, $\affX=\aff(\set X_1,\set X_2)$, 
we have
\begin{equation}\label{eqn:affinity-preserving}
  \left|\affY^2-\affX^2\right|\le C(d_1-\affX^2)\delta, 
\end{equation}
where $C>0$ is some universal constant. 
\end{theorem}
\begin{remark}
Consequently, 
the projection Frobenius-norm distance of $\set X_1,\set X_2$ is preserved by $\mphi$:
\begin{equation}\label{eqn:subspace-rip}
  \left|D^2(\set Y_1,\set Y_2)-D^2(\set X_1,\set X_2)\right|\le C\delta D^2(\set X_1,\set X_2),
\end{equation}
which means that $\mphi$ possesses \emph{subspace RIP} for subspaces of $\set X$.
\end{remark}
\begin{remark}
  The assumption \eqref{eqn:orthonormal-perturbation-bound} 
  has a intimate connection with the concept of \emph{subspace embedding} 
  in numerical linear algebra \cite{Clarkson2017Low}. 
  The main difference is that subspace embedding in \cite{Clarkson2017Low} 
  asks \eqref{eqn:orthonormal-perturbation-bound} 
  to hold with probability at least $1-\varepsilon$, 
  hence is a probabilistic assumption, 
  while in our assumption \eqref{eqn:orthonormal-perturbation-bound} is deterministic 
  and removes the need of probabilistic argument; 
  we believe our assumption better captures the essence of the matter.
\end{remark}

The assumption \eqref{eqn:orthonormal-perturbation-bound} is equivalent to 
saying that $\mphi$ acts as a near-isometry on $\set X$, 
i.e. $(1-\delta)\|\u\|^2<\|\mphi\u\|^2<(1+\delta)\|\u\|^2$ for any $\u\in\set X$. 
Thus the essence of Theorem \ref{thm:cov-est-to-subspace-rip} is that 
a near-isometry on a subspace $\set X$ preserves the pairwise distance of subspaces of $\set X$. 
This is not an obvious fact, 
since affinity, hence subspace distance, 
is defined in a subtle way that involves orthonormal bases of both subspaces, 
but the latter is not preserved by a near-isometry.
The overall effect of such structural degeneration 
makes the desired factor $(d_1-\affX^2)$, 
which is crucial in establishing \eqref{eqn:subspace-rip}, out of immediate reach.
One has to perform some careful analysis to obtain \eqref{eqn:affinity-preserving} and \eqref{eqn:subspace-rip}. 

Before we present the proof (in Section \ref{sec:proof}), 
it is of interest to explain how Theorem \ref{thm:cov-est-to-subspace-rip} 
leads easily to a series of corollaries 
on subspace RIP for a wide variety of random matrices, 
which we will do in the next section.
}

% ---------------------------------------------------------

{\color{black}\section{Random Matrices and Near-Isometry on Subspaces}\label{sec:examples}
This section discusses in detail 
the near-isometry condition \eqref{eqn:orthonormal-perturbation-bound} 
and its connection with random matrices. 
Furthermore, this section examines various random matrices encountered in practice 
and shows that they satisfy the near-isometry condition, 
hence possess subspace RIP, 
which would validate their application in subspace-related tasks to accelerate computation.

The near-isometry condition \eqref{eqn:orthonormal-perturbation-bound} 
and Theorem \ref{thm:cov-est-to-subspace-rip} are best understood in the context of random matrices. 
In practice, it is useless to discuss a pair of low-dimensional subspaces $\set X_1,\set X_2$ contained in 
a \emph{specific} subspace $\set X$; 
one would often need \eqref{eqn:affinity-preserving} and \eqref{eqn:subspace-rip} 
for \emph{any} pair of such subspaces. 
For any pair of subspaces $\set X_1$, $\set X_2$, 
we consider their sum 
\[\set X=\{\x_1+\x_2:\x_1\in\set X_1,\x_2\in\set X_2\}.\]
This is a subspace of $\mathbb R^N$ of dimension at most $2d_2$ that contains both $\set X_1$ and $\set X_2$. 
If $\mphi$ acts as a near-isometry on $\set X$, 
then $\mphi$ preserves the affinity and the distance between $\set X_1$, $\set X_2$. 
In order that $\mphi$ preserves the affinity and the distance between any pair of $\set X_1$, $\set X_2$, 
one may impose that $\mphi$ acts as a near-isometry on any subspace of dimension $2d_2$ in $\mathbb R^N$.
This is, however, apparently impossible for deterministic $\mphi$, 
and a standard way to resolve this is to use a random matrix $\mphi$ instead. 

It is clear from the above argument that, 
if $\mphi$ is a random matrix which acts as a near-isometry with high probability 
on any subspace of dimension $2d_2$, 
then $\mphi$ preserves the affinity and the distance between $\set X_1$, $\set X_2$ 
with high probability for any pair of $\set X_1$, $\set X_2$ of dimension 
respectively $d_1$, $d_2$, where $d_1\le d_2$. 

As a consequence, analysis of subspace RIP now boils down to 
analysis of the singular values of $\mphi\U$,
where $\U$ is a matrix whose columns constitute an orthonormal basis for some subspace in $\mathbb R^N$.
This will be carried out in the rest of this section.

Throughout this section, $\set X_1,\ldots,\set X_L$ always denote subspaces in $\mathbb R^N$ 
of dimension $d_1,\ldots,d_L$, and $d_*=\max\{d_1,\ldots,d_L\}$. 
The image of $\set X_i$ under the random projection $\mphi$ will be denoted by $\set Y_i$. 
The subspace RIP of $\mphi$ will be characterized by maximum discrepancy
\begin{equation*}
  \Delta=\max_{1\le i<j\le L}\frac{|\aff^2(\set Y_i,\set Y_j)-\aff^2(\set X_i,\set X_j)|}{\max\{d_i,d_j\}-\aff^2(\set X_i,\set X_j)}.
\end{equation*}
Note that 
\begin{equation*}
  \Delta\ge\max_{1\le i<j\le L}\frac{|D^2(\set Y_i,\set Y_j)-D^2(\set X_i,\set X_j)|}{D^2(\set X_i,\set X_j)}.
\end{equation*}
}

{\color{black}
\subsection{Example: Exponential Johnson-Lindenstrauss Property}}
A class of random projections that deserves much emphasis is 
the ones with exponential Johnson-Lindenstrauss property\footnote{In literature the same property is usually under the name ``Johnson-Lindenstrauss property'', 
without ``exponential''; see for instance \cite{Foucart2017Mathematical}, Section 9.5.},
defined as following 
\begin{definition}\label{def:jl-property}
A random matrix $\A\in\mathbb R^{n\times N}$ is said to 
satisfy \emph{exponential Johnson-Lindenstrauss property}, 
if there exists some constant $\tilde c>0$, 
such that for any $0<\varepsilon<1$ 
and for any $\x\in\mathbb R^N$, 
\begin{equation*}
    \probP(\left|\|\A\x\|^2-\|\x\|^2\right|>\varepsilon\|\x\|^2)\le2{\rm e}^{-\tilde c\varepsilon^2n}.
\end{equation*} 
\end{definition}

Examples of random matrices with exponential Johnson-Lindenstrauss property 
are pervasive in both theory and practice. 
Section \ref{apd:jl-property} provides a non-comprehensive list of such examples 
(Gaussian matrices with independent columns, subgaussian matrices with independent rows, 
and partial Fourier/Hadamard matrices)
and also a related theorem which asserts that classical RIP for sparse vectors 
with sufficiently small restricted isometry constants implies 
exponential Johnson-Lindenstrauss property. 

Taking the route discussed at the beginning of this section, 
we have
\begin{lemma}\label{lem:JL-orthonormal-preserving}
Let $\U$ be a matrix whose columns constitute an orthonormal basis 
for a $d$-dimensional subspace in $\mathbb R^N$. 
Assume the random matrix $\mphi$ satisfies exponential Johnson-Lindenstrauss property. 
Then for any $0<\varepsilon<1$, 
we have
\begin{equation*}
\probP(1-\varepsilon<s_{\min}^2(\mphi\U)\le s_{\max}^2(\mphi\U)<1+\varepsilon)\ge 1-{\rm e}^{-\tilde c\varepsilon^2 n+3d}.
\end{equation*}
\end{lemma}
The proof is by a standard covering argument and is deferred to Section \ref{apd:jl-property}. 

As a corollary, we have the following result on subspace RIP of random matrices 
with exponential Johnson-Lindenstrauss property, 
which generalizes the main result in \cite{Li2019Rigorous}.
\begin{corollary}\label{cor:JL-implies-SRIP}
Assume the random matrix $\mphi$ satisfies exponential Johnson-Lindenstrauss property. 
Then for some universal constant $c>0$ and for any $0<\varepsilon<1$, 
we have
\[
  \Delta\le\varepsilon
\]
with probability at least $1-L^2{\rm e}^{-c\tilde c\varepsilon^2n+6d_*}$. 
In particular, whenever $n>24c^{-1}\tilde c^{-1}\varepsilon^{-2}\max\{d_*,\log L\}$, 
the probability is at least $1-\mathrm e^{-c\tilde c\varepsilon^2n/2}$.
\end{corollary}
\begin{proof}
  By the argument at the beginning of this section, 
  this follows from Lemma \ref{lem:JL-orthonormal-preserving}, Theorem \ref{thm:cov-est-to-subspace-rip} 
  and union bound.
\end{proof}

Note that how this simple proof supersedes, in both effectivity and generality, 
the complicated probabilistic analysis which spans tens of pages in \cite{Li2019Rigorous}, 
thanks to Theorem \ref{thm:cov-est-to-subspace-rip}. 
We will discuss this difference in more detail in Section \ref{sec:discussions}.

Corollary \ref{cor:JL-implies-SRIP} permits
to apply various matrices used in practice (e.g. Bernoulli or partial Fourier) to subspace related tasks. 
For subgaussian matrices, the constant $\tilde c$ in exponential Johnson-Lindenstrauss property depends only 
on the subgaussian norm and is inverse proportional to the square of the subgaussian norm, 
which is quite satisfying.
However, for partial Fourier matrices the above analysis is a bit rough, 
as the constant $\tilde c$ is proportional to $\sqrt N$. 
This is problematic when $N$ is large\footnote{In \cite{Heckel2017Dimensionality} 
it is claimed that $\tilde c=\Omega(\log^4 N)$, 
which was not legitimately proved there and is likely wrong. 
In fact, they argue that $\tilde c=\Omega(\log^4 N)$ follows from Theorem \ref{thm:apd:rip-to-jl} 
and Theorem \ref{thm:apd:rip-of-bos} presented in our Section \ref{apd:jl-property}, 
but to achieve the $\mathrm e^{-\Omega(n)}$ probability bound one has to take $s=\Omega(n)$ 
and $\zeta=\mathrm e^{-\Omega(n)}$, which requires that $n=\Omega(n^2)$ and is absurd.
}. 
Moreover, this leaves out the commonly-used partial circulant matrices and partial Toeplitz matrices, 
which do not satisfy exponential Johnson-Lindenstrauss property with reasonable $\tilde c$. 
These matrices are endowed with fast matrix-vector multiplication algorithms 
that significantly accelerate the random compression procedure, 
hence are worth a more refined treatment, 
as shown in the next example.

{\color{black}
\subsection{Example: Some Random Matrices With Fast Algorithms}}
Some random matrices are particularly fascinating for practical use 
due to their advantages in computational efficiency 
and their natural emergence in signal processing tasks. 
Such examples include partial Fourier matrices 
and other randomly sampled Bounded Orthonormal Systems (BOS) \cite{Foucart2017Mathematical}, 
which correspond to random subsampling in frequency domain and other feature domains. 
Another important example is partial circulant/Toeplitz matrix, 
which corresponds to subsampling after a random convolution. 
These matrices allow for $O(N\log N)$-time multiplication-by-vector algorithms 
by virtue of Fast Fourier Transform (FFT), Fast Walsh-Hadamard Transform (FWHT), etc. 

Proving subspace RIP of these matrices would legitimate their use in subspace related tasks, 
hence significantly improves the efficiency in handling such tasks. 
In fact, we will show in Section \ref{sec:simulations} how the application of random compression 
by these matrices boosts up subspace clustering on a real-world dataset. 
This motivates the following results. 
\begin{lemma}\label{lem:BOS-orthonormal-preserving}
Let $\U$ be a matrix whose columns constitute an orthonormal basis 
for a $d$-dimensional subspace in $\mathbb R^N$. 
Let $\A\in\mathbb C^{n\times N}$ be the random sampling associated 
to a BOS\footnote{Partial Fourier matrices and partial Hadamard matrices 
are both randomly samplings associated to a BOS with constant $O(1)$.} with constant $K\ge 1$. 
Let $\D_{\epsilon}$ be a diagonal matrix with i.i.d. Rademacher random variables on its diagonal. 
Let $\mphi=\A\D_\epsilon$. 
Then there exists some constant $C>1$ such that 
for any $\varepsilon\in(0,1)$ and for any $n>CK^3\varepsilon^{-3}\max\{d\log^3d,\log^3N\}$,
we have \[1-\varepsilon<s_{\min}^2(\mphi\U)\le s_{\max}^2(\mphi\U)<1+\varepsilon\]
with probability at least
\[1-\exp\left(-C^{-1}\big(\sqrt{d^2+K^{-2}\varepsilon^2n}-d\big)\right).\]
\end{lemma}
The proof is by a careful application of well-known properties of randomly sample BOS 
and is deferred to Section \ref{apd:fast-matrices}. 
Note that the exponent $3$ in $K^3$, $\varepsilon^{-3}$, $\log^3 d$ and $\log^3N$ 
can be replaced by $2+\epsilon$ for any $\epsilon>0$, 
and we chose $3$ only for typographical convenience.

\begin{corollary}\label{cor:BOS-implies-SRIP}
Let $\mphi$ be as in Lemma \ref{lem:BOS-orthonormal-preserving}. 
Then there exists some constant $C>1$ such that for any $\varepsilon\in(0,1)$
and any $n>CK^3\varepsilon^{-3}\max\{d_*(\log^3 d_*+\log L),\log^2L,\log^3N\}$, 
we have
\[
  \Delta\le\varepsilon
\]
with probability at least 
\[1-\exp\left(-C^{-1}\big(\sqrt{d_*^2+K^{-2}\varepsilon^2n}-d_*\big)\right).\]
\end{corollary}
\begin{proof}
By the argument at the beginning of this section, 
this follows from Lemma \ref{lem:BOS-orthonormal-preserving}, Theorem \ref{thm:cov-est-to-subspace-rip} 
and union bound.
\end{proof}

While Lemma \ref{lem:BOS-orthonormal-preserving} follows from standard results on randomly sampled BOS, 
for partial circulant/Toeplitz matrices the situation is more subtle. 
The following lemma will be proved in Section \ref{apd:fast-matrices}
via some modifications of the proof strategy in \cite{Vybiral2011Variant}.
\begin{lemma}\label{lem:circulant-orthonormal-preserving}
Let $\U$ be a matrix whose columns constitute an orthonormal basis 
for a $d$-dimensional subspace in $\mathbb R^N$. 
Let $\vect a$ be a random vector in $\mathbb R^N$ 
with i.i.d. standard complex circular Gaussian entries. 
Let $\mathcal C(\vect a)$ be the circulant matrix generated by $\vect a$, 
i.e. whose first row is $\vect a$.
Choose arbitrarily $n$ rows of $\mathcal C(\vect a)$ 
and form with these rows a new matrix $\A\in\mathbb R^{n\times N}$. 
Let $\mphi=\frac{1}{\sqrt n}\A$. 
Then there exists some constant $C>1$ such that
for any $\varepsilon\in(0,1)$ 
and any $n>C\varepsilon^{-2}\max\{d,\log N\}^2$, 
we have \[1-\varepsilon<s_{\min}^2(\mphi\U)\le s_{\max}^2(\mphi\U)<1+\varepsilon\]
with probability at least
$1-\mathrm e^{-c\varepsilon\sqrt n}$.
\end{lemma}
\begin{remark}\label{rem:subspace-embedding}
  The above lemma requires that $n=O(d_*^2)$.
  With a substantial amount of work
  (using some modern results on generic chaining bound of suprema of order-$2$ chaos process, 
  e.g. \cite{Dirksen2015Tail}),
  it is possible to attain a sub-optimal scaling similar to the one in Corollary \ref{cor:BOS-implies-SRIP}, 
  as well as a similar probability bound as in Corollary \ref{cor:BOS-implies-SRIP}. 
  An easier way to improve the scaling is to utilize the result in \cite{Rauhut2012Restricted}, 
  which allows to obtain $n=O(d_*^{3/2})$ with a trade-off in the probability bound
  that becomes $\mathrm e^{-O(n^{1/3})}$; 
  this turns out to be a special case of the aforementioned "harder" treatment.
  We will not pursue these directions here to avoid unnecessary technicality.
\end{remark}
\begin{corollary}\label{cor:circulant-implies-SRIP}
  Let $\mphi$ be as in Lemma \ref{lem:circulant-orthonormal-preserving}. 
  Then there exists some constant $C>1$ such that
  for any $\varepsilon\in(0,1)$ and any $n>C\varepsilon^{-2}\max\{d_*,\log N,\log L\}^2$, we have
  \[
    \Delta\le\varepsilon
  \]
  with probability at least $1-\mathrm e^{-c\varepsilon\sqrt n}$. 
\end{corollary}
\begin{proof}
By the argument at the beginning of this section, 
this follows from Lemma \ref{lem:circulant-orthonormal-preserving}, Theorem \ref{thm:cov-est-to-subspace-rip} 
and union bound.
\end{proof}
\begin{remark}
  The same is true for Toeplitz matrix since a Toeplitz matrix can be embedded into a circulant matrix
  with twice dimension. 
\end{remark}
\begin{remark}
  In fact, we will prove the Lemma \ref{lem:circulant-orthonormal-preserving} 
  for random vector $\vect a$ 
  with independent complex uniformly-subgaussian entries
  (see Section \ref{apd:jl-property}, Definition \ref{def:subgaussian}). 
  This involves modifying and generalizing the proof in \cite{Vybiral2011Variant}.
\end{remark}

{\color{black}
\subsection{Example: Heavy-Tailed Distributions}}
In practice one may also have to deal with heavy-tailed random matrices. 
Subspace RIP of heavy-tailed random matrices 
is now handy by Theorem \ref{thm:cov-est-to-subspace-rip} 
and standard results in covariance estimation. 
Before presenting these results, we need to set up 
some customary assumptions on the rows of $\mphi$. 
Denote the rows of $\mphi$ 
by $\frac{1}{\sqrt n}\mathbf x_1^{\rm T},\ldots,\frac{1}{\sqrt n}\mathbf x_n^{\rm T}$. 
\begin{enumerate}
  \item $\x_i$'s are centered, i.e. $\mathbb E\x_i=\mathbf 0$.
  \item $\x_i$'s are isotropic, i.e. $\mathbb E\x_i\x_i^{\mathrm T}=\mathbf I$.
  \item $\x_i$'s are independent.
\end{enumerate}
These assumptions (centered, isotropic and independent rows) are quite natural 
and often serve as the default setting in compressed sensing 
and non-asymptotic random matrix theory, 
especially in context of Bai-Yin law; 
see \cite{Vershynin2010Introduction} for examples and further discussions. 
Here we briefly mention that the a row of $\mphi$ can be regarded as 
a linear functional that observes the data vector, 
and independence of rows means that different observations are independent, 
which is reasonable in many applications. 
On the other hand, centeredness and isotropy 
can be fulfilled by a preprocessing step before observing the data, 
thus are not really restrictive. 
Some results in random matrix theory are also valid for weakly-independent rows, 
but such results usually bear considerable technicality 
imposed by the difficulty to quantitively define weak dependence, 
hence are not discussed here.

We will deal with two important types of heavy-tailed distributions: 
those with finite moments, and log-concave ensembles.  
\paragraph{Finite moments} 
  Let $\eta>1$, $C'\ge1$ be constants and $\x$ be an $N$-dimensional random
  vectors which is centered and isotropic. 
  The random vector $\x$ is said to 
  satisfy the strong regularity condition \cite{Srivastava2013Covariance} if  
  \begin{equation}\label{eqn:strong-regularity}
    \probP(\|\projP\x\|^2>t)\le C't^{-\eta},\quad\text{for $t>C'\operatorname{rank}\projP$}.
  \end{equation}
  for every orthogonal projection $\projP$ of rank at most $d$ in $\mathbb R^N$, 
  where $C$ is some universal constant.
  This condition is satisfied, for example, by those $\x$ whose entries are independent 
  and have uniformly bounded $(4+\varepsilon)$-moments. 
  We will discuss the meaning of this condition later.
  \begin{lemma}\label{lem:strong-regularity-orthonormal-preserving}
    Let $\U$ be a matrix whose columns constitute an orthonormal basis 
    for a $d$-dimensional subspace in $\mathbb R^N$. 
    Assume $\x_1,\ldots,\x_n$ are independent centered isotropic random vectors in $\mathbb R^N$ 
    that satisfy the strong regularity condition \eqref{eqn:strong-regularity}. 
    Let $\mphi$ be a $n\times N$ random matrix whose rows are 
    $\frac1{\sqrt n}\x_1^{\mathrm T},\ldots,\frac1{\sqrt n}\x_n^{\mathrm T}$.
    Then there exists a polynomial function $\mathrm{poly}(\cdot)$ 
    whose coefficients depend only on $\eta$ and $C'$, 
    such that whenever $\varepsilon\in(0,1)$ and $n>\mathrm{poly}(\varepsilon^{-1})d$, 
    we have
    \[1-\varepsilon<s_{\min}^2(\mphi\U)\le s_{\max}^2(\mphi\U)<1+\varepsilon\]
    with probability at least $1-\varepsilon$.
  \end{lemma}
  \begin{proof}[Sketch]
    Set $\y_i=\U^{\mathrm T}\x_i$, 
    then it is easy to check that $\y_i$'s are independent, centered and isotropic, 
    and 
    \[\|\U^{\mathrm T}\mphi^{\mathrm T}\mphi\U-\vect I\|=\left\|\frac1n\sum_{i=1}^n\y_i\y_i^{\mathrm T}-\vect I\right\|.\]
    One may verify that $\y_i$ satisfies strong regularity condition
    if $\x_i$ satisfies strong regularity condition, 
    and the main theorem in \cite{Srivastava2013Covariance} yields the desired results. 
    Details are postponed to Section \ref{apd:heavy-tails}.
  \end{proof}
  \begin{corollary}\label{cor:strong-regularity-SRIP}
    Let $\mphi$ be as in Lemma \ref{lem:strong-regularity-orthonormal-preserving}.
    Then there exists a polynomial function $\mathrm{poly}(\cdot)$ 
    whose coefficients depend only on $\eta$ and $C'$, 
    such that whenever $\varepsilon\in(0,1)$ and $n>\mathrm{poly}(\varepsilon^{-1})Ld_*$, we have
    \begin{equation*}
      \Delta\le\varepsilon
    \end{equation*}
    with probability at least $1-\varepsilon$.
  \end{corollary}
  \begin{proof}
    Note that this does \emph{not} follow from union bound. 
    Instead, we take $\set X$ to be the sum of all $\set X_i$, 
    then $\dim\set X\le Ld_*$.
    Now the conclusion follows from Lemma \ref{lem:strong-regularity-orthonormal-preserving} 
    (applied to an orthonormal basis of $\set X$)
    and Theorem \ref{thm:cov-est-to-subspace-rip}. 
  \end{proof}
  \begin{remark}
    A recent result \cite{Xu2019Convergence} of some of the authors show that 
    the partial strong regularity condition 
    can be further relaxed to that \eqref{eqn:strong-regularity}
    holds for any $\projP$ with rank $\lceil Cd_*\rceil$. 
    Though this assumption is logically weaker, 
    it does not seem to yield sensible improvement in our case.
  \end{remark}

  An intriguing feature of these results is that strong regularity condition, 
  which in its original context seems to be an artifact 
  due to the spectral sparsfier method adopted in \cite{Srivastava2013Covariance}, 
  arises very naturally in our case. 
  In fact, to control the singular values of $\mphi\U$ 
  in Lemma \ref{lem:strong-regularity-orthonormal-preserving}, 
  it is necessary to have a reasonably fast tail-decay for $\|\y_i\|$ 
  (see \cite{Bai1988Note, Srivastava2013Covariance, Tikhomirov2017Sample} for discussions), 
  but $\|\y_i\|$ is just the norm of the orthogonal projection of $\x_i$ onto 
  the arbitrary $d$-dimensional subspace spanned by the columns of $\U$; 
  thus it is necessary to assume that the tails of all orthogonal projections of $\x_i$ 
  decay fast enough, 
  which is exactly what is meant by strong regularity condition.

  If $\x$ has independent entries, 
  \eqref{eqn:strong-regularity} is satisfied when its entries have uniformly bounded $4\eta$-th moments. 
  (For instance, see Proposition 1.3 in \cite{Srivastava2013Covariance}). 
  Hence one may regard strong regularity condition as a finite $4+\epsilon$ moment assumption. 
  Note that it has been known for long that finite fourth moments 
  are necessary for covariance estimation \cite{Bai1988Note}.

\paragraph{Log-concave ensembles}
  If the row vectors of $\mphi$ have better tail behavior, 
  stronger probability bounds can be obtained. 
  A class of distributions with heavy, but not too heavy tails 
  that plays a role in geometric functional analysis is log-concave distribution. 
  A probability distribution $\probP$ on $\mathbb R^N$ is said to be log-concave if 
  \[\probP(\theta A+(1-\theta)B)\ge\probP(A)^\theta\probP(B)^{1-\theta}\]
  for any measurable sets $A$, $B$ in $\mathbb R^N$ and any $\theta\in[0,1]$. 
  It follows from definition that the marginal of a log-concave distribution is again log-concave. 
  Typical examples of log-concave distributions include 
  the uniform distribution on a convex body (e.g. a Euclidean ball)
  or more generally, the distribution with density $C\exp(-f(\x))$, 
  where $f$ is a convex function\footnote{By setting $f$ to be the indicator function 
  of a convex body we recover the case of uniform distribution on a convex body.}. 
  This subsumes Laplace, Gaussian, Gamma, Beta, Weibull 
  and Logistic distributions (in suitable region of parameters).

  \begin{lemma}\label{lem:log-concave-orthonormal-preserving}
    Let $\U$ be a matrix whose columns constitute an orthonormal basis 
    for a $d$-dimensional subspace in $\mathbb R^N$. 
    Assume $\x_1,\ldots,\x_n$ are independent random vectors in $\mathbb R^N$ 
    that are centered, isotropic and log-concave. 
    Let $\mphi$ be a $n\times N$ random matrix whose rows are
    $\frac1{\sqrt n}\x_1^{\mathrm T},\ldots,\frac1{\sqrt n}\x_n^{\mathrm T}$. 
    Then there exists some universal constants $c>0$ and $C>1$ such that 
    for any $\varepsilon\in(0,1)$ and any $n>C\varepsilon^{-2}d$, 
    we have \[1-\varepsilon<s_{\min}^2(\mphi\U)\le s_{\max}^2(\mphi\U)<1+\varepsilon\]
    with probability at least $1-\mathrm e^{-c\varepsilon\sqrt n}$.
  \end{lemma}
  This follows from a tricky application of 
  a famous theorem from \cite{Adamczak2010Quantitative,Adamczak2011Sharp}. 
  Details are postponed to Section \ref{apd:heavy-tails}.
  \begin{corollary}\label{cor:log-concave-SRIP}
    Let $\mphi$ be as in Lemma \ref{lem:log-concave-orthonormal-preserving}
    Then there exists some universal constants $c>0$ and $C>0$ such that 
    whenever $\varepsilon\in(0,1/2)$ and $n>C\varepsilon^{-2}\max\{d_*, \log^2L\}$, 
    we have 
    \begin{equation}
      \Delta\le\varepsilon
    \end{equation}
    with probability at least $1-\mathrm e^{-c\varepsilon\sqrt n}$.
  \end{corollary}
  \begin{proof}
    By the argument at the beginning of this section, 
    this follows from Lemma \ref{lem:log-concave-orthonormal-preserving}, Theorem \ref{thm:cov-est-to-subspace-rip} 
    and union bound.
  \end{proof}
\begin{remark}
The $e^{-\Omega(\sqrt n)}$ probability bound is optimal
due to thin shell probability of log-concave ensembles, 
see \cite{Guedon2014Concentration}.
\end{remark}

% ---------------------------------------------------------
\section{Proof of Theorem \ref{thm:cov-est-to-subspace-rip}}\label{sec:proof}

\subsection{Some First Consequences of \eqref{eqn:orthonormal-perturbation-bound}}
Recall that \eqref{eqn:orthonormal-perturbation-bound} means 
$\mphi$ acts as a near-isometry on $\set X$. 
An immediate consequence of this assumption is:
\begin{proposition}\label{prop:1d-orthogonal-preserving}
Let $\set X$ be a subspace of $\mathbb R^N$ 
and $\U$ be a matrix whose columns constitute an orthonormal basis of $\set X$. 
Let $\mphi$ be a $n\times N$ matrix such that \eqref{eqn:orthonormal-perturbation-bound} holds 
for some $\delta\in(0,1)$. 
Then we have 
\begin{equation*}
  \sqrt{1-\delta}\|\u\|\le\|\mphi\u\|\le\sqrt{1+\delta}\|\u\|
\end{equation*}
for any $\u\in\set X$. 
In particular, if $\delta\in(0,1/4)$ we have $\sqrt{\frac34}\|\u\|\le\|\mphi\u\|\le\sqrt{\frac54}\|\u\|$.

Moreover, assume $\set X_1$, $\set X_2$ are two subspaces of $\set X$ 
which are orthogonal to each other, 
and let $\U_1$ (resp. $\U_2$) be a matrix 
whose columns constitute an orthonormal basis of $\set X_1$ (resp. $\set X_2$), 
then
\begin{equation*}
  \|\U_2^{\mathrm T}\mphi^{\mathrm T}\mphi\U_1\|\le\delta.
\end{equation*}
\end{proposition}
\begin{proof}
For any $\u\in\set X$, one may find a vector $\x$ in $\mathbb R^{\dim\set X}$ 
such that $\u=\U\x$, hence $\|\u\|=\|\x\|$. Thus $\|\mphi\u\|=\|\mphi\U\x\|$ 
is at least $s_{\min}(\mphi\U)\|u\|$ and at most $s_{\max}(\mphi\U)\|\u\|$. 
This proves the first part of the proposition. 

For the second part, note that
\begin{align*}
  \|\U_2^{\mathrm T}\mphi^{\mathrm T}\mphi\U_1\|
  &=\sup_{\substack{\x_1\in\mathbb R^{\dim\set X_1}\\ \x_2\in\mathbb R^{\dim\set X_2}}}\x_2^{\mathrm T}\U_2^{\mathrm T}\mphi^{\mathrm T}\mphi\U_1\x_1\\
  &=\sup_{\substack{\u_1\in\set X_1, \|\u_1\|=1\\ \u_2\in\set X_2,\|\u_2\|=1}}\u_2^{\mathrm T}\mphi^{\mathrm T}\mphi\u_1. 
\end{align*}
But
\begin{equation*}
  \u_2^{\mathrm T}\mphi^{\mathrm T}\mphi\u_1=\frac14(\|\mphi(\u_1+\u_2)\|^2-\|\mphi(\u_1-\u_2)\|^2). 
\end{equation*}
the conclusion follows from the first part and the fact that $\|\u_1+\u_2\|^2=\|\u_1-\u_2\|^2=2$ 
(since $\set X_1$, $\set X_2$ are orthogonal to each other).
\end{proof}

\begin{lemma}\label{lem:line-aff-orthogonal}
Under the same setting as in Lemma \ref{thm:cov-est-to-subspace-rip}, 
assume further that $\set X_1$, $\set X_2$ are orthogonal to each other. 
Then we have
\begin{equation}\label{eqn:in-lem-line-aff-orthogonal1}
  \|\ptwo\mphi\u\|\le\sqrt{\frac43}\|\u\|\delta<\frac43\|\mphi\u\|\delta
\end{equation}
for any $\u\in\set X_1$. 
\end{lemma}
\begin{proof}
Let $\U_1$ be an orthonormal basis for $\set X_1$
and $\U_2$ be an orthonormal basis for $\set X_2$, 
thus $\U_1$ (resp. $\U_2$) is a $N\times d_1$ (resp. $N\times d_2$) matrix with orthonormal columns.
Set $\V_1=\mphi\U_1$, $\V_2=\mphi\U_2$. 
By \eqref{eqn:orthonormal-perturbation-bound}, 
$\V_2$ is of full rank and all of its singular values are in 
$(\sqrt{1-\delta}, \sqrt{1+\delta})$.
In this case we have 
\begin{equation*}
  \ptwo=\V_2(\V_2^{\rm T}\V_2)^{-1}\V_2^{\rm T}.
\end{equation*}
Thus 
\begin{align*}
  \|\ptwo\V_1\|^2&=\V_1^{\rm T}\V_2(\V_2^{\rm T}\V_2)^{-1}\V_2^{\rm T}\V_1\\
  &\le s_{\min}^{-2}(\V_2)\|\V_2^{\rm T}\V_1\|^2\\
  &\le (1-\delta)^{-1}\|\U_2^{\rm T}\mphi^{\rm T}\mphi\U_1\|^2\\
  &\le (1-\delta)^{-1}\delta^2,
\end{align*}
where the last inequality follows from Proposition \ref{prop:1d-orthogonal-preserving}. 
For any $\u\in\set X_1$, it is possible to find some $\x\in\mathbb R^d_1$ 
such that $\u=\U_1\x$, hence $\|\u\|=\|\x\|$;
we thus have
\begin{equation*}
  \|\ptwo\mphi\u\|=\|\ptwo\mphi\U_1\x\|=\|\ptwo\V_1\x\|\le\|\ptwo\V_1\|\|\u\|.
\end{equation*}
Hence $\|\ptwo\mphi\u\|\le(1-\delta)^{-1/2}\delta\|\u\|\le\sqrt{\frac43}\|\u\|\delta$. 
Along with Proposition \ref{prop:1d-orthogonal-preserving}, 
this proves \eqref{eqn:in-lem-line-aff-orthogonal1}. 
\end{proof}

\subsection{One-dimensional Subspaces}
\begin{lemma}\label{lem:uniform-line-SRIP}
Under the same setting as in Theorem \ref{thm:cov-est-to-subspace-rip}, 
we have
\begin{equation}\label{eqn:in-lem-uniform-line-SRIP}
    \left|\aff^2(\set Y_2,\mphi\u)-\aff^2(\set X_2,\u)\right|\le C\left(1-\aff^2(\set X_2,\u)\right)\delta
\end{equation}
for all $\u\in\set X_1$, $\u\ne\mathbf0$. 
\end{lemma}
\begin{proof}
The proof is by straightforward computation. 
It suffices to prove \eqref{eqn:in-lem-uniform-line-SRIP} for unit vectors $\u_1\in\set X_1$. 
For any such unit vector, there exists some unit vector $\u_{2}\in\set X_2$, 
such that $\lambda\overset{\triangle}{=}\langle\u_1,\u_{2}\rangle=\aff(\set X_2,\u)$
by Lemma \ref{lem:affinity-compute}). 
In fact, $\u_{2}$ is 
the direction vector of the projection of $\u$ onto $\set X_2$. 
We thus have 
\begin{equation}\label{line_decomposition}
    \u_1=\lambda\u_{2}+\sqrt{1-\lambda^2}\u_0,
\end{equation}
where $\u_0\in\set X$ is some unit vector orthogonal to $\set X_2$. 
Recall that the squared affinity of $\set Y_1$ and $\set Y_2$ is defined as
\begin{equation*}
    \affYS=\frac{1}{\|\mphi\u_1\|^2}\|\ptwo\mphi\u_1\|^2.
\end{equation*}

In light of (\ref{line_decomposition}), we have
\begin{align}
    \|\mphi\u_1\|^2= &~\lambda^2\|\mphi\u_{2}\|^2 + (1-\lambda^2)\|\mphi\u_0\|^2 \nonumber\\ 
    &  + 2\lambda\sqrt{1-\lambda^2}\langle\mphi\u_{2},\mphi\u_0\rangle,\nonumber\\
    \|\ptwo\mphi\u_1\|^2 
    = &~\lambda^2\|\ptwo\mphi\u_{2}\|^2 + (1-\lambda^2)\|\ptwo\mphi\u_0\|^2\nonumber \\ 
    &  + 2\lambda\sqrt{1-\lambda^2}\langle\ptwo\mphi\u_{2},\ptwo\mphi\u_0\rangle.\label{proj_norm_expansion}
\end{align}

Note that $\ptwo\mphi\u_{2}=\mphi\u_{2}$ 
since $\mphi\u_{2}\in\mphi\set X_2=\set Y_2$. 
Furthermore, 
\begin{align*}
\langle\ptwo\mphi\u_{2},\ptwo\mphi\u_0\rangle&=\langle\ptwo\mphi\u_{2},\mphi\u_0\rangle\\
&=\langle\mphi\u_{2},\mphi\u_0\rangle,
\end{align*}
where the first equality is elementary geometry. 
Thus 
\begin{align*}
\|\ptwo\mphi\u_1\|^2
= &~\lambda^2\|\mphi\u_{2}\|^2 + (1-\lambda^2)\|\ptwo\mphi\u_0\|^2 \\ 
&  + 2\lambda\sqrt{1-\lambda^2}\langle\mphi\u_{2},\mphi\u_0\rangle.
\end{align*}

Combining these equations, we have
\begin{align*}
\left|\affYS - \lambda^2\right|=\frac{(1-\lambda^2)}{\|\mphi\u_1\|^2}
\Big(&\lambda^2(\|\mphi\u_{2}\|^2-\|\mphi\u_0\|^2) \\ 
        &+\|\ptwo\mphi\u_0\|^2 \\
        &+2\lambda\sqrt{1-\lambda^2}\langle\mphi\u_{2},\mphi\u_0\rangle\Big). 
\end{align*}

Since $\|\u_1\|=\|\u_2\|=\|\u_0\|=1$ and that $\u_0$ is perpendicular to $\set X_2$ (hence to $\u_2$), 
the above quantity is bounded by $C(1-\lambda^2)\delta$ 
by Proposition \ref{prop:1d-orthogonal-preserving} and Lemma \ref{lem:line-aff-orthogonal}. 
This completes the proof.
\end{proof}

\subsection{The General Case}
Now we prove the full version of Theorem \ref{thm:cov-est-to-subspace-rip}.
Choose a principal orthonormal bases (Lemma \ref{lem:affinity-compute}) 
$\U_1$, $\U_2$ for $\set X_1, \set X_2$. 
In this proof we also borrow the notation of $\lambda_k$ from Lemma \ref{lem:affinity-compute}.
Denote $\V_1=\mphi\U_1$, $\V_2=\mphi\U_2$. 
The $k$-th column of $\U_1$, $\U_2$ and $\V_1$ are respectively 
denoted by $\u_{1,k}$, $\u_{2,k}$ and $\v_{1,k}$. 
Note that $\v_{1,k}=\mphi\u_{1,k}$ by definition. 

By \eqref{eqn:orthonormal-perturbation-bound},
$\V_1$, $\V_2$ are of full rank 
and all of their singular values lie in $(\sqrt{1-\delta}, \sqrt{1+\delta})$.
We shall need two auxiliary matrices derived from $\V_1$. 
The first one is the column-normalized version of $\V_1$, 
defined as $\hat{\V}_1=[\frac{\v_{1,1}}{\|\v_{1,1}\|},\ldots,\frac{\v_{1,d_1}}{\|\v_{1,d_1}\|}]$. 
The second one is the orthogonal matrix obtained from Gram-Schmidt orthogonalization of columns of $\V_1$, 
which we denote by $\Q_1$. 
The $k$-th column of $\hat{\V}_1$ and $\Q_1$ 
are respectively denoted by $\hat{\v}_{1,k}$ and $\q_{1,k}$.
Let $\ptwo$ be the orthogonal projection onto $\set Y_2$, i.e. the column space of $\V_2$. 
We have
\begin{align}
    \affY^2-\affX^2=&\|\ptwo\Q_1\|_{\rm F}^2-\|\U_2^{\rm T}\U_1\|_{\rm F}^2\nonumber\\
    =&(\|\ptwo\Q_1\|_{\rm F}^2-\|\ptwo\hat{\V}_1\|_{\rm F}^2)\nonumber\\
    &+(\|\ptwo\hat{\V}_1\|_{\rm F}^2-\|\U_2^{\rm T}\U_1\|_{\rm F}^2).\label{eqn:space_aff_telescope}
\end{align}

We estimate the last two quantities in (\ref{eqn:space_aff_telescope}) respectively.
\begin{proposition}
We have
\begin{equation*}
  \left|\|\ptwo\hat{\V}_1\|_{\rm F}^2-\|\U_2^{\rm T}\U_1\|_{\rm F}^2\right|\le C(d_1-\affX^2)\delta.
\end{equation*}
\end{proposition}
\begin{proof}
Note that
\begin{equation}\label{eqn:decomposition1}
  \|\ptwo\hat{\V}_1\|_{\rm F}^2-\|\U_2^{\rm T}\U_1\|_{\rm F}^2=\sum_{k=1}^{d_1}\left(\|\ptwo\hat{\v}_{1,k}\|^2-\|\U_2^{\rm T}\u_{1,k}\|^2\right).
\end{equation}

Observe that $\|\U_2^{\rm T}\u_{1,k}\|$ is the affinity between $\set X_2$ 
and the one-dimensional subspace spanned by $\u_{1,k}$, 
while $\|\ptwo\hat{\v}_{1,k}\|$ is the affinity between $\set Y_2$ 
and the one-dimensional subspace spanned by $\mphi\u_{1,k}$. 
Furthermore, $\|\U_2^{\rm T}\u_{1,k}\|=\lambda_k$ 
since $\U_1$, $\U_2$ are principal orthonormal bases.
By Lemma \ref{lem:line-aff-orthogonal}, we have
\begin{equation}\label{eqn:error-of-normalized-col}
  \left|\|\ptwo\hat{\v}_{1,k}\|^2-\|\U_2^{\rm T}\u_{1,k}\|^2\right|\le C(1-\lambda_k^2)\delta.
\end{equation}

Summing up, we obtain
\begin{align*}
    \left|\|\ptwo\hat{\V}_1\|_{\rm F}^2-\|\U_2^{\rm T}\U_1\|_{\rm F}^2\right|&\le C(d_1-\sum_{k=1}^{d_1}\lambda_k^2)\delta\\
    &=C(d_1-\affX^2)\delta,
\end{align*}
as desired.
\end{proof}

\begin{proposition}
We have
\begin{equation*}
    \left|\|\ptwo\Q_1\|_{\rm F}^2-\|\ptwo\hat{\V}_1\|_{\rm F}^2\right|\le C(d_1-\affX^2)\delta.
\end{equation*}
\end{proposition}
\begin{proof}
Similar to (\ref{eqn:decomposition1}), we have 
\begin{equation*}
    \|\ptwo\Q_1\|_{\rm F}^2-\|\ptwo\hat{\V}_1\|_{\rm F}^2=\sum_{i=1}^{d_1}(\|\ptwo\q_{1,k}\|^2-\|\ptwo\hat{\v}_{1,k}\|^2).
\end{equation*}

Denote by $\set Z_k$ the space spanned by $\v_{1,1},\ldots,\v_{1,k}$.
Then 
\begin{equation}\label{eqn:gram-schmidt}
    \q_{1,k}=\frac{\hat{\v}_{1,k}-\pz{k-1}\hat{\v}_{1,k}}{\|\hat{\v}_{1,k}-\pz{k-1}\hat{\v}_{1,k}\|}.
\end{equation}

Note that $\q_{1,k}$, $\hat{\v}_{1,k}$ are unit vectors, hence by Pythagorean theorem we have
\begin{equation}\label{eqn:pytha1}
    \|\hat{\v}_{1,k}-\pz{k-1}\hat{\v}_{1,k}\|^2=1-\|\pz{k-1}\hat{\v}_{1,k}\|^2,
\end{equation}
and
\begin{align}
    &\|\ptwo\q_{1,k}\|^2-\|\ptwo\hat{\v}_{1,k}\|^2\nonumber\\
    =&(1-\|\ptwop\q_{1,k}\|^2)-(1-\|\ptwop\hat{\v}_{1,k}\|^2)\nonumber\\
    =&\|\ptwop\hat{\v}_{1,k}\|^2-\|\ptwop\q_{1,k}\|^2,\label{eqn:pytha2}
\end{align}
where $\set Y_2^\perp$ denotes the orthogonal complement of $\set Y_2$. 
Combining (\ref{eqn:gram-schmidt}) and (\ref{eqn:pytha1}) we obtain
\begin{align*}
    \|\ptwop\q_{1,k}\|^2=&\frac{\|\ptwop\hat{\v}_{1,k}-\ptwop\pz{k-1}\hat{\v}_{1,k}\|^2}{1-\|\pz{k-1}\hat{\v}_{1,k}\|^2}\\
    =&\phantom{+}\frac{\|\ptwop\hat{\v}_{1,k}\|^2}{1-\|\pz{k-1}\hat{\v}_{1,k}\|^2}\\
    &-\frac{2\langle\ptwop\hat{\v}_{1,k},\ptwop\pz{k-1}\hat{\v}_{1,k}\rangle}{1-\|\pz{k-1}\hat{\v}_{1,k}\|^2}\\
    &+\frac{\|\ptwop\pz{k-1}\hat{\v}_{1,k}\|^2}{1-\|\pz{k-1}\hat{\v}_{1,k}\|^2}.
\end{align*}

This together with (\ref{eqn:pytha2}) yields
\begin{align}
    &\|\ptwo\q_{1,k}\|^2-\|\ptwo\hat{\v}_{1,k}\|^2\nonumber\\
    =&-\frac{\|\ptwop\hat{\v}_{1,k}\|^2\|\pz{k-1}\hat{\v}_{1,k}\|^2}{1-\|\pz{k-1}\hat{\v}_{1,k}\|^2}\nonumber\\
    &+\frac{2\langle\ptwop\hat{\v}_{1,k},\ptwop\pz{k-1}\hat{\v}_{1,k}\rangle}{1-\|\pz{k-1}\hat{\v}_{1,k}\|^2}\nonumber\\
    &-\frac{\|\ptwop\pz{k-1}\hat{\v}_{1,k}\|^2}{1-\|\pz{k-1}\hat{\v}_{1,k}\|^2}.\label{eqn:decomposition2}
\end{align}

Since $\u_{1,k}$ is perpendicular to the subspace spanned by 
$\u_{1,1},\ldots,\u_{1,k-1}$, 
by Lemma \ref{lem:line-aff-orthogonal}, 
$\|\pz{k-1}\hat{\v}_{1,k}\|^2\le\frac{16}9\delta^2<\frac49\delta<1/2$. 
The proof would be complete once we show
\begin{align}
    \|\ptwop\pz{k-1}\hat{\v}_{1,k}\|^2&\le C(1-\lambda_k^2)\delta^2,\label{eqn:component1}\\
    \|\ptwop\hat{\v}_{1,k}\|^2&\le C(1-\lambda_k^2).\label{eqn:component2}
\end{align}

By (\ref{eqn:error-of-normalized-col}) and the discussion prior to it 
(which says that $\|\U_2^{\rm T}\u_{1,k}\|=\lambda_k$),
the following holds:
\begin{align*}
    \|\ptwop\hat{\v}_{1,k}\|^2&=1-\|\ptwo\hat{\v}_{1,k}\|^2\\
    &=(1-\lambda_k^2)-(\|\ptwo\hat{\v}_{1,k}\|^2-\|\U_2^{\rm T}\u_{1,k}\|^2)\\
    &\le(1-\lambda_k^2)(1+C\delta)\\
    &\le C(1-\lambda_k^2),
\end{align*}
where the first inequality follows from Lemma \ref{lem:uniform-line-SRIP} 
applied to $\set X_2$ and the $1$-dimensional subspace spanned by $\u_{1,k}$, 
noting that $\|\ptwo\hat{\v}_{1,k}\|$ is the affinity between 
$\set Y_2$ and $\mphi\u_{1,k}$.
This establishes \eqref{eqn:component2}. 
For \eqref{eqn:component1}, however, some more work is required. 
Let $\Z_{k-1}$ be a orthonormal basis of $\set Z_{k-1}$. 
Then $\pz{k-1}=\Z_{k-1}\Z_{k-1}^{\rm T}$, which implies 
\begin{align*}
    \|\ptwop\pz{k-1}\hat{\v}_{1,k}\|^2&=\|\ptwop\Z_{k-1}\Z_{k-1}^{\rm T}\hat{\v}_{1,k}\|^2\\
    &\le s_{\max}^2(\ptwop\Z_{k-1})\|\Z_{k-1}^{\rm T}\hat{\v}_{1,k}\|^2.
\end{align*}

One recognizes at once that $\|\Z_{k-1}^{\rm T}\hat{\v}_{1,k}\|^2=\|\pz{k-1}\hat{\v}_{1,k}\|^2$, 
which is bounded by $C\delta^2$ according to Lemma \ref{lem:line-aff-orthogonal}. 
It remains to prove 
\begin{equation}\label{eqn:last-step}
    s_{\max}^2(\ptwop\Z_{k-1})\le C(1-\lambda_k^2).
\end{equation}

Since $\Z_{k-1}$ has orthonormal columns, it follows that
\begin{align}
    s_{\max}^2(\ptwop\Z_{k-1})&=\sup_{\x\in\mathbb S^{k-2}}\|\ptwop\Z_{k-1}\x\|^2\nonumber\\
    &=\sup_{\substack{\|\x\|=1\\ \x\in\set Z_{k-1}}}\|\ptwop\x\|^2.\label{eqn:s_max-conversion}
\end{align}

By definition, $\set Z_{k-1}$ is spanned by the first $(k-1)$ columns of $\V_1=\mphi\U_1$. 
Denote by $\U_{1,1:k-1}$ the first $(k-1)$ columns of $\U_1$. 
We have
\begin{align*}
    \sup_{\substack{\|\x\|=1\\ \x\in\set Z_{k-1}}}\|\ptwop\x\|^2
    &=1-\inf_{\substack{\|\x\|=1\\ \x\in\set Z_{k-1}}}\|\ptwo\x\|^2\\
    &=1-\inf_{\x\in\mathbb S^{k-2}}\frac{\|\ptwo\mphi\U_{1,1:k-1}\x\|^2}{\|\mphi\U_{1,1:k-1}\x\|^2}\\
\end{align*}
Note that 
\[\frac{\|\ptwo\mphi\U_{1,1:k-1}\x\|^2}{\|\mphi\U_{1,1:k-1}\x\|^2}=\aff^2(\set Y_2,\mphi\U_{1,1:k-1}\x).\]
Applying Lemma \ref{lem:uniform-line-SRIP} to $\set X_2$ 
and the $1$-dimensional subspace spanned by $\U_{1,1:k-1}\x$, 
we obtain
\begin{align}
  \sup_{\substack{\|\x\|=1\\ \x\in\set Z_{k-1}}}\|\ptwop\x\|^2
  &\le (1-\|\projP_{\set X_2}\U_{1,1:k-1}\x\|^2)(1+C\delta)\nonumber\\
  &\le C(1-\|\projP_{\set X_2}\U_{1,1:k-1}\x\|^2)\label{eqn:compression-last-step}
\end{align}
Finally we need to bound $\|\projP_{\set X_2}\U_{1,1:k-1}\x\|^2$.
Since $\projP_{\set X_2}\u_{1,i}=\lambda_i\u_{2,i}$ 
by choice of principal orthonormal bases, 
we have 
\begin{align}
    \|\projP_{\set X_2}\U_{1,1:k-1}\x\|^2=\sum_{i=1}^{k-1}\lambda_i^2x_i^2&\ge\sum_{i=1}^{k-1}\lambda_k^2x_i^2\nonumber\\
    &=\lambda_k^2.\label{eqn:eig-estimation-for-last-step}
\end{align}
Combining \eqref{eqn:s_max-conversion}, \eqref{eqn:compression-last-step}
and \eqref{eqn:eig-estimation-for-last-step},
we obtain \eqref{eqn:last-step} as desired.
\end{proof}

% ------------------------------------------------------

\section{Discussions}\label{sec:discussions}
This section discusses the proofs of Theorem \ref{thm:cov-est-to-subspace-rip} 
and other results in this paper, 
and their differences from previous works.

{\color{black}
A notable feature of Theorem \ref{thm:cov-est-to-subspace-rip} 
is being purely deterministic, i.e. does not involve randomness of $\mphi$. 
This cleans up the fog spanned 
across the mixture of probabilistic arguments and deterministic arguments 
in previous works on dimensionality-reduced 
subspace clustering \cite{Heckel2017Dimensionality, Wang2019Theoretical}
and on subspace RIP \cite{Li2019Rigorous}. 
Note that the concept of subspace embedding (see Remark \ref{rem:subspace-embedding})
that is somehow similar to \eqref{eqn:orthonormal-perturbation-bound} 
appears in the analysis in \cite{Wang2019Theoretical}. 
Theorem \ref{thm:cov-est-to-subspace-rip} is distinguished from theirs at least in two aspects: 
our analysis is deterministic and does not involve randomness, 
while the assumption in \cite{Wang2019Theoretical} is an probabilistic inequality; 
we proceed to analyze the subspace RIP of $\mphi$, 
which is capable of handling various subspace-related tasks and algorithms,
while \cite{Wang2019Theoretical} considered only SSC, 
a special algorithm of the specific problem of subspace clustering. 
Also note that the result in \cite{Wang2019Theoretical} requires $n=\Omega(d^{9/2})$ 
for subgaussian matrices to ensure success of CSC, 
which is worse than the optimal scaling $n=O(d)$ that can be obtained by our theory, 
for instance, using the framework in \cite{Meng2018CSC}. 

Readers may find some similarity between the proof of Theorem \ref{thm:cov-est-to-subspace-rip} 
and the proof in \cite{Li2019Rigorous}. 
Indeed, the routine computations in both proofs are the same, 
e.g. \eqref{proj_norm_expansion}, \eqref{eqn:space_aff_telescope}--\eqref{eqn:decomposition2}. 
However, these routine computations are merely a non-substantial part of the proof, 
and the core difficulty is how to bound the quantities involved in these equations. 
To this end, our proof deviate significantly from that in \cite{Li2019Rigorous}. 
In fact, \cite{Li2019Rigorous} relied heavily on the fact that the Gaussian projection 
of two orthogonal vectors are independent to bound terms such as $\|\ptwo\V_1\|$ in Lemma \ref{lem:line-aff-orthogonal}
and $\|\ptwo\mphi\u_0\|$ in \eqref{proj_norm_expansion}, 
which cannot be generalized to even subgaussian matrices, 
let alone partial Fourier matrices, partial circulant matrices or log-concave ensembles,
and the situation got more involved when the quantity of interest is complicated: 
it took over four pages in \cite{Li2019Rigorous} (see Appendix 8.8 there) 
to obtain the sought-for conditional independence 
to bound $\|\ptwop\pz{k-1}\hat\v_{1,k}\|$ in \eqref{eqn:component1}.
In our proof, we propose the novel Lemma \ref{lem:uniform-line-SRIP} 
(which was proved in \cite{Li2019Rigorous} but only for one-dimensional $\set X_1$) 
as a key intermediate step, 
and all the bounds needed follow easily from either assumption \eqref{eqn:orthonormal-perturbation-bound} 
or Lemma \ref{lem:uniform-line-SRIP}. 
This demonstrates the power of our abstract setting \eqref{eqn:orthonormal-perturbation-bound}. 
}

% ------------------------------------------------------

\section{Applications}\label{sec:applications}
In this section we briefly mention some of the applications of our theory. 
As we mentioned before, 
our theory provides a universal framework to analyze the effects of random compression 
on subspace related tasks, 
and providing a list of such tasks would be out of scope of this paper. 
However, it is possible to describe the universal framework which 
works for any subspace-related algorithms that admits a theoretical guarantee via affinity, 
as follows.\\
\textbf{Framework. }\emph{Assume that we have a theoretical guarantee for an algorithm $A$ 
that succeeds on a colletion of $L$ subspaces $\set X_1,\ldots,\set X_L$ 
with probability at least $1-\delta$. 
Let $\mphi$ be the random matrix described in 
Corollary \ref{cor:JL-implies-SRIP}, 
\ref{cor:BOS-implies-SRIP}, 
\ref{cor:circulant-implies-SRIP}, 
\ref{cor:strong-regularity-SRIP} or \ref{cor:log-concave-SRIP}. 
Let $\varepsilon$ be a sufficiently small positive number that depends 
only on the relative position of subspaces. 
Then for $n$ satisfying the restriction in the corresponding corollary, 
algorithm $A$ succeeds on the data projected by $\mphi$ 
with probability at least $1-\delta_{\text{proj}}-\delta$, 
where $\delta_{\text{proj}}$ is the error probability given in the corresponding corollary.
}

As examples, we roughly describe some useful consequences of our theory on two tasks: 
subspace clustering and active subspace detection \cite{Lodhi2018Detection}. 
\begin{theorem}[Compressed subspace clustering]\label{thm:sc}
  Let $\mphi$ be a partial Fourier matrix.
  Under some technical assumptions on the algorithm parameters
  (that is irrelevant of $\mphi$, see \cite{Meng2018CSC}), 
  the Threshold-based Subspace Clustering (TSC) algorithm 
  succeeds (see \cite{Meng2018CSC}) on the dataset compressed by $\mphi$ with probability at least 
  \[1-\frac{10}M-\sum_{l=1}^L\mathrm e^{-c(N_l+1)}-\mathrm e^{-c(\sqrt{d_*^2+K^{-2}\varepsilon^2n}-d_*)}\]
  given $n>C\varepsilon^{-3}\max\{d_*(\log^3d_*+\log L),\log^2L,\log^3N\}$ and 
  \begin{align*}
    &\max _{k \ne l} \ \sqrt{ \frac{{\aff}^2(\set X_k, \set X_l)(1-\varepsilon) 
    + d_*\varepsilon }{d_k \wedge d_l} } \\
    &+ 6\sqrt{\frac{d_*}{n}}\left(1+\sqrt{\frac{6\log {M}}{d_*}}\right)^2 \\
    \leq & \frac{1}{\sqrt{6M_{\max}\log {M}}},
  \end{align*}
  where $c>0$ is a constant, $M$ denotes the number of data points 
  and $M_{\max}$ denotes the maximal number of data points lying in the same subspace.
\end{theorem}
This follows from Corollary \ref{cor:circulant-implies-SRIP} the analysis scheme proposed in \cite{Meng2018CSC}. 
We have chosen partial Fourier matrix and TSC algorithm only for simplicity of presentation; 
similar results hold for subgaussian matrices, partial circulant matrices, etc., and 
for SSC, SSC-OMP algorithms, etc. 
It is perhaps worth mentioning that Theorem \ref{thm:sc} is, 
to the best of our knowledge, 
the first effective performance analysis with sub-optimal scaling $n=O(d_*\text{polylog}(d_*))$ 
for TSC compressed by partial Fourier matrices, 
since the analysis in \cite{Heckel2017Dimensionality} appears flawed 
as pointed out in a footnote of this paper.

We now turn to the task of active subspace detection. 
\begin{theorem}[Compressed active subspace detection]
  Let $\mphi$ be a partial Fourier matrix. 
  The compressed maximum-likelihood detector (see \cite{Jiao2019Compressed}) for 
  noiseless active subspace detection 
  succeeds, under Gaussian assumption on the data distribution, 
  with probability at least 
  \[
    1-4\sum_{i\ne j}\mathrm e^{-K(\aff_*(1-\varepsilon)+d\varepsilon)d_*}-\mathrm e^{-c(\sqrt{d_*^2+K^{-2}\varepsilon^2n}-d_*)},
  \]
  where 
  \[
    K(x):=\frac18\frac{(1-x/d_*-8/d_*)^2}{4+(1-x/d_*-8/d_*)}
  \]
  and $\aff_*$ denotes $\max_{k\ne l}\aff(\set X_k,\set X_l)$, 
  given $\varepsilon\in(0,1)$ and $n>C\varepsilon^{-3}\max\{d_*(\log^3d_*+\log L),\log^2L,\log^3N\}$. 
  For noisy case a similar conclusion holds.
\end{theorem}
This follows from Corollary \ref{cor:BOS-implies-SRIP} and the analysis scheme 
proposed in \cite{Jiao2019Compressed}. 
Note that in \cite{Jiao2019Compressed} the above theorem is proved for random matrices with exponential 
Johnson-Lindenstrauss property (Theorem 5 there) using results from the first version of this paper. 
Again, the appearance of partial Fourier matrices is arbitrary and can be replaced by any other random matrices 
discussed in this paper.

% ------------------------------------------------------

\section{Simulations}\label{sec:simulations}
We verify our results on Yale Face Database B \cite{Georghiades2001Few}
and test the performance of Sparse Subspace Clustering (SSC) after random projection 
by Gaussian matrix, partial Fourier/Hadamard matrix, partial circulant matrix, 
and matrix with i.i.d. Student-t distributed ($\nu=5$) entries. 
Yale Face Database B is a database with ambient dimension $N=32256$ that contains the face images 
of $10$ human subjects. 
For convenience we randomly select $4$ subjects whose face images are subsequently clustered. 
The matrices we chose are representatives of the three classes of random matrices which we have inspected: 
Gaussian matrix represents matrices with exponential Johnson-Lindenstrauss property, 
partial Fourier/Hadamard matrix and partial circulant matrix represent matrices with fast algorithms, 
and i.i.d. Student-t distributed matrices, which has infinite fifth moments, represents matrices with heavy tails.

Performance is evaluated in terms of clustering error rate \cite{Elhamifar2013Sparse}, 
i.e. the rate that SSC algorithm clusters a randomly compressed image to the correct subject,
see Fig. \ref{fig:err-rate}. 
We are also concerned with the boost-up in computational efficiency 
supplied by fast matrix-vector multiplication algorithms 
for partial Fourier/Hadamard matrices and partial circulant matrices, 
which will be evaluated in terms of average running time, 
i.e. the time it takes to compute the random projection of a high-dimensional vector, 
see Fig. \ref{fig:running-time}. 

\paragraph{Computational Complexity}
For unstructured $n\times N$ random matrices such as Gaussian matrices and Student-t matrices, 
it takes $\complexityO(nN)$ time to compute the random projection of a vector. 
For partial Fourier/Hadamard matrices and partial circulant matrices, 
$\complexityO(N\log N)$-time algorithms exist, 
thanks to Fast Fourier Transform (FFT) and Fast Walsh-Hadamard Transform (FWHT). 
More precisely, one may compute the sign-randomized\footnote{
    This means multiplying each entry of $\x$ by a Rademacher random variable. 
    See Theorem \ref{thm:apd:rip-to-jl} for details.
} version of a vector $\x\in\mathbb R^N$ 
in $\complexityO(N)$ time, 
and then compute its fast Fourier/Walsh-Hadamard transform $\hat{\x}$ in $\complexityO(N\log N)$ time.
By randomly sampling $n$ entries from $\hat{\x}$, which takes $\complexityO(n)$ time, 
one finally obtains the randomly projected version of $\x$. 
For typical scenarios in practice we have $n\gg\log N$, 
hence random projections by partial Fourier/Hadamard matrices and partial circulant matrices 
are much efficient than random projections by unstructured matrices. 

\paragraph{Discussions}
As one may see in Fig. \ref{fig:err-rate}, 
the error rates of all types of random projections converge to the baseline, 
i.e. the error rate of SSC without random projection, as $n$ tends to $N$. 
This is consistent with our theory that all these random matrices 
preserve affinities between subspaces. 
The running time of random projection shown in Fig. \ref{fig:running-time} 
coincides with our analysis above. 
Note that the running time is plotted in logarithmic scale.
For partial Fourier/Hadamard matrix and partial circulant matrix,
the running time is almost constant in $n$. 
Both the running times of Student-t matrix and Gaussian matrix grow linearly with respect to $n$;
this is because that they both involve a $\complexityO(nN)$-time matrix-by-vector multiplication.
Note that Student-t matrix takes a somewhat longer time than Gaussian matrix, 
possibly due to higher complexity in its implementation, 
i.e. in generating Student-t distributed variables.
The running time of partial Hadamard matrix is longer than 
that of partial Fourier matrix and partial circulant matrix, 
which may be caused by a less efficient implementation of FWHT than that of FFT. 
Except for very small $n$, 
partial Fourier/Hadamard matrix and partial circulant matrix is significantly faster than 
Gaussian matrix and Student-t matrix.
For small $n$ the running time of Gaussian/Student-t matrix is shorter than 
that of partial Hadamard matrix. 
However, as $n$ grows large, for instance when $n>10000$, 
partial Hadamard matrix becomes the better choice.
Our analysis indicates that this advantage would become even more obvious 
when the ambient dimension $N$ is larger and $n\gg\log N$.

\begin{figure}\centering
\includegraphics[width=0.7\linewidth]{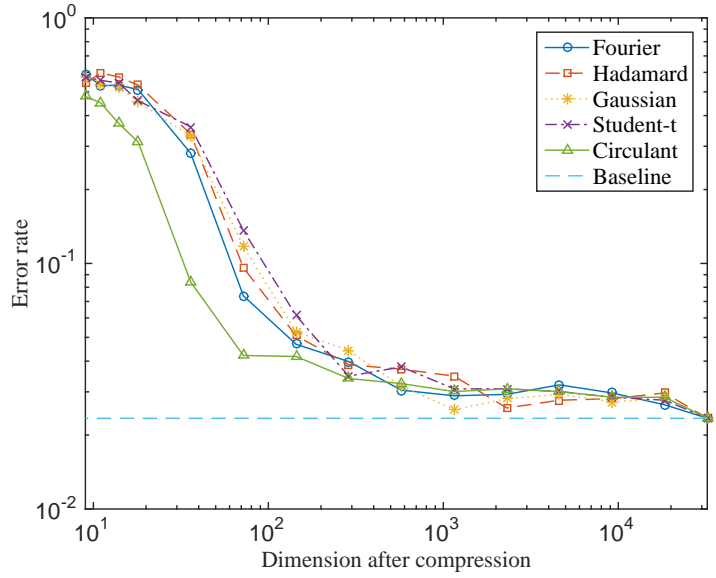}
\caption{Clustering error rate vs. compressed dimension $n$ for Yale Face Dataset B. 
The error rate of SSC without random compression is approximately 2.3\%.
}\label{fig:err-rate}
\end{figure}

\begin{figure}\centering
\includegraphics[width=0.7\linewidth]{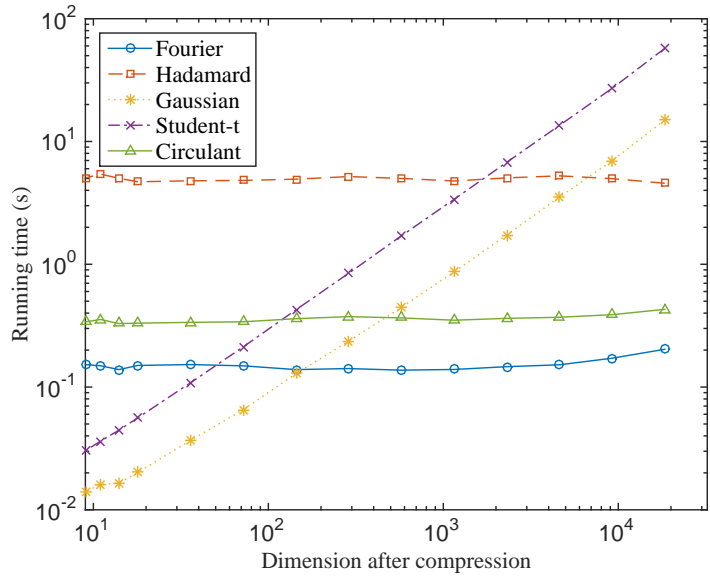}
\caption{Average running time of random projecting $256$ vectors vs. compressed dimension $n$. 
Note that for partial Fourier matrix and partial Hadamard matrix the running time is almost 
independent of $n$. 
Both $x$-axis and $y$-axis are drawn in logarithmic scale, 
so the running time of Gaussian matrix and Bernoulli matrix grows linearly in the figure.}
\label{fig:running-time}
\end{figure}

% ------------------------------------------------------

\section{Conclusion}\label{sec:conclusion}
In this paper we provided a deterministic characterization of subspace RIP 
in terms of near-isometry on subspaces. 
This result enables to analyze the subspace RIP of matrices with a unified approach.
As examples, we prove with this result that a large variety of random matrices, 
including subgaussian matrices, partial Fourier/Hadamard matrices and partial circulant/Toeplitz matrices, 
random matrices with independent strongly regular rows, and log-concave ensembles. 
This significantly enlarges the collection of random matrices known to possess subpace RIP 
in literature, demonstrating the applicability of subspace RIP. 

Subspace RIP, or in plain language, 
the almost-invariance of affinity under random projections, 
has played an important role in the analysis 
of Compressed Subspace Clustering algorithms. 
Hence our result demonstrates more scenarios where random projection and CSC may apply 
and has the potential to give better performance guarantee for CSC algorithms. 
Furthermore, 
since subspace RIP is a universal concept that does not depend on any specific algorithm, 
our result may find its application in various subspace-based machine learning algorithms, 
which we leave to future research.

% ------------------------------------------------------

\section{Appendix: Exponential Johnson-Lindenstrauss Property}\label{apd:jl-property}

This appendix deals with random matrices with exponential Johnson-Lindenstrauss property 
and their subspace RIP. 
The appendix is divided into two parts. 
In the first part, we prove that subgaussian matrices and partial Fourier/Hadamard matrices
satisfy exponential Johnson-Lindenstrauss property 
to illustrate the wide applicability of this concept. 
In the second part, we introduce the standard tool of covering arguments 
and use it to prove Lemma \ref{lem:JL-orthonormal-preserving}.
\subsection{Examples of Random Matrices with Exponential JL Property}
Here we provide a non-comprehensive list of common random matrices that fulfill
exponential Johnson-Lindenstrauss property. 
Such matrices can be roughly divided into two categories, whose exponential Johnson-Lindenstrauss property 
respectively stems from subgaussian concentration property
and Restricted Isometry Property for sparse vectors. 
The most important example in the first category is subgaussian random matrices, 
and that in the second category is randomly sampled Bounded Orthogonal Systems (BOS), 
a class of random matrices including partial Fourier matrices and partial Hadamard matrices.
We discuss these two categories respectively. 
\paragraph{Subgaussian concentration}
We begin by defining subgaussian random variables and subgaussian random vectors. 
\begin{definition}[Section 7.1, \cite{Foucart2017Mathematical}]\label{def:subgaussian}
Fix a positive constant $K$. 
A \emph{$K$-subgaussian random variable} is a random variable $X$ satisfying 
\begin{equation*}
    \mathbb E{\rm e}^{tx}\le {\rm e}^{K^2t^2/2},
\end{equation*}
for any $t\in\mathbb R$. 
A \emph{subgaussian matrix} is a random matrix 
with each entry a $K$-subgaussian random variable.
\end{definition}

Gaussian variables might be the most common examples of subgaussian random variables. 
Other examples include variables with Rademacher distribution 
or uniform distribution on $[-1,1]$. 
In fact, any centered bounded random variable is subgaussian. 

It is possible to generalized the above definition 
to a multi-dimensional setting.
\begin{definition}
Let $\bm\Gamma$ be a positive semidefinite matrix.
An \emph{$\bm\Gamma$-subgaussian random vector} is a random vector $\u$ 
taking value in $\mathbb R^n$ such that for any $\x\in\mathbb R^n$, 
\begin{equation*}
    \mathbb E{\rm e}^{\langle\u,\x\rangle}\le {\rm e}^{\langle\bm\Gamma\x,\x\rangle/2}.
\end{equation*}
Such a random vector is said to satisfy 
\emph{Bernstein condition}, if 
\begin{equation*}
    \mathbb E\left|\|\u\|^2-\mathbb E\|\u\|^2\right|^k\le Ck!\|\bm\Gamma\|_{\rm op}^{k-2}\|\mathbb E(\u\u^{\rm T})\|_{\rm F}^2
\end{equation*}
for some constant $C>0$, where $\|\cdot\|_{\rm op}$ denotes the operator norm.
\end{definition}

\begin{theorem}[\cite{Chen2018Hanson}, Theorem 2.10]
If $\A$ is a random matrix with independent 
$\bm\Gamma$-subgaussian columns satisfying Bernstein condition. 
Furthermore, assume $\mathbb E(\A^{\rm T}\A)={\bf I}$.
Then $\A$ satisfies exponential Johnson-Lindenstrauss property. 
\end{theorem}

\begin{corollary}
Assume $\A$ is a random matrix satisfying $\mathbb E(\A^{\rm T}\A)={\bf I}$. 
If in addition $\A$ is in one of the following form, 
then $\A$ satisfies exponential Johnson-Lindenstrauss property: 
\begin{enumerate}[label=\alph*)]
    \item A random matrix with independent subgaussian rows;
    \item A Gaussian matrix with independent columns;
    \item A product of a positive semidefinite matrix 
    and a subgaussian matrix with independent entries.
\end{enumerate}
\end{corollary}
\begin{proof}
    a) is classical and can be found in \cite{Vershynin2010Introduction}; 
    b) and c) are proved in \cite{Chen2018Hanson}.
\end{proof}

\paragraph{Restricted Isometry Property}
Restricted isometry property for sparse vectors \cite{Candes2008Restricted}
has been a very powerful tool in analysis of compressed sensing 
and related algorithms. 
A vector is called $s$-sparse if at most $s$ of its entries are non-zero. 
RIP for sparse vectors are defined in the following way:
\begin{definition}
A matrix $\A$ is said to possess the RIP if 
there exists a function $\delta(s)\ge0$, 
such that for any positive integer $s$ and any $s$-sparse vector $\x$,
\begin{equation*}
    (1-\delta(s))\|\x\|^2\le\|\A\x\|^2\le(1+\delta(s))\|\x\|^2.
\end{equation*} 
The function $\delta(s)$ is called the \emph{restricted isometry constant} 
of $\A$.
\end{definition}

It is easy to see that exponential Johnson-Lindenstrauss property implies RIP, 
see for instance Theorem 5.2 in \cite{Baraniuk2008Simple}). 
The converse is also true in some sense, as the following theorem shows.
\begin{theorem}[\cite{Rauhut2010Compressive,Krahmer2011New}]\label{thm:apd:rip-to-jl}
Assume $\A$ is an $n\times N$ matrix with RIP and 
restricted isometry constant $\delta(s)$. 
Fix some $\varepsilon\in(0,1/2)$. 
Assume further that for some $s>0$ we have $\delta(s)<\varepsilon/4$.
Let $\D_{\epsilon}$ be a diagonal matrix with i.i.d. Rademacher random variables 
on its diagonal, then for any $\x\in\mathbb R^N$, we have
\begin{equation*}
    \probP(\left|\|\A\D_{\epsilon}\x\|^2-\|\x\|^2\right|>\varepsilon\|\x\|^2)\le 2{\rm e}^{-\tilde c s}
\end{equation*}
for some universal constant $\tilde c>0$.
\end{theorem}

Partial Fourier matrices and partial Hadamard matrices 
are both examples of a more general class of random matrices, 
namely random sampled Bounded Orthonormal Systems (BOS). 
For such matrices it was shown that their restricted isometry constants 
are sufficiently small:
\begin{theorem}[\cite{Rauhut2010Compressive, Foucart2017Mathematical}]\label{thm:apd:rip-of-bos}
Let $\A\in\mathbb C^{n\times N}$ be the random sampling associated 
to a BOS with constant $K\ge 1$. 
For $\zeta,\eta_1,\eta_2\in(0,1)$, if
\begin{align*}
    \frac{n}{\log(9n)}&\ge C_1\eta_1^{-2}K^2s\log^2(4s)\log(8N),\\
    n&\ge C_2\eta_2^{-2}K^2s\log(\zeta^{-1}),
\end{align*}
then with probability at least $1-\zeta$ the restricted isometry 
constant $\delta(s)$ of $\frac{1}{\sqrt n}\A$ satisfies 
$\delta(s)\le\eta_1+\eta_1^2+\eta_2$. 
(Here $C_1$, $C_2$ are universal positive constants.)
\end{theorem}

One may combine Theorem \ref{thm:apd:rip-to-jl} and \ref{thm:apd:rip-of-bos} 
to obtain several modified versions of exponential Johnson-Lindenstrauss property that 
random sampled BOS satisfy. 
For example, taking $\eta_1=\eta_2=\varepsilon/4$, $s=\lceil C\varepsilon^2n/\sqrt N\rceil$ and $\zeta={\rm e}^{-C's}$, 
one may obtain that 
\begin{equation*}
    \probP(\left|\|\A\D_{\epsilon}\x\|^2-\|\x\|^2\right|>\varepsilon\|\x\|^2)\le 2{\rm e}^{- c\varepsilon^2n/\sqrt N}.
\end{equation*}

\subsection{Proof of Lemma \ref{lem:JL-orthonormal-preserving}}
We will use some standard covering arguments, e.g. \cite{Baraniuk2008Simple},
to prove \eqref{eqn:orthonormal-perturbation-bound} for 
random matrices with exponential Johnson-Lindenstrauss property.
\begin{definition}
An \emph{$\varepsilon$-net} in a subset $X$ of a Euclidean space is a finite subset $\mathcal N$ 
of $X$ such that for any $x\in X$ we have
\begin{equation*}
    \min_{z\in\mathcal N}\|x-z\|<\varepsilon.
\end{equation*}
The \emph{metric entropy} of $X$ is a function $N(X,\varepsilon)$
defined as the minimum cardinality of an $\varepsilon$-net of $X$.
\end{definition}

For subsets of Euclidean space, 
the metric entropy can be easily bounded by a volume packing argument. 
For the Euclidean unit ball the corresponding result reads as following:
\begin{lemma}[Proposition C.3, \cite{Foucart2017Mathematical}]\label{lem:metric-entropy-bound}
Let $B_n$ be the unit ball in $\mathbb R^n$. Then
\begin{equation*}
    N(B_n, \varepsilon)\le\left(1+\frac{2}{\varepsilon}\right)^n.
\end{equation*}
\end{lemma}

The usage of covering arguments is demonstrated by the following lemma:
\begin{lemma}[\cite{Vershynin2010Introduction}, Lemma 5.3]\label{lem:covering-approximation}
Suppose $\mathcal N$ is a $\frac12$-net of $\mathbb S^{n-1}$.
Let $\A$ be a $n\times n$ matrix. Then 
\begin{equation*}
  \|\A\|\le 2\sup_{\x\in\mathcal N}\|\A\x\|.
\end{equation*}
\end{lemma}

Now we are ready to finish the proof of Lemma \ref{lem:JL-orthonormal-preserving}.
Note that it suffices to show 
\begin{equation}\label{eqn:in-cor-orthonormal-preserving1}
    \probP\left(\max_{\x\in\mathbb S^{d-1}}\left|\|\mphi\U\x\|^2-1\right|>\varepsilon\right)\le {\rm e}^{-\tilde c\varepsilon^2 n+3d}.\
\end{equation}
For any $\x\in\mathbb S^{d-1}$, exponential Johnson-Lindenstrauss property implies
\begin{equation}\label{eqn:in-cor-orthonormal-preserving2}
    \probP(\left|\|\mphi\U\x\|^2-1\right|>\varepsilon)\le 2{\rm e}^{-\tilde c\varepsilon^2 n}.
\end{equation}
The desired inequality (\ref{eqn:in-cor-orthonormal-preserving1}) 
follows from (\ref{eqn:in-cor-orthonormal-preserving2}) 
and a standard covering argument. 
By Lemma \ref{lem:metric-entropy-bound}, one may find a set $\mathcal N\subseteq\mathbb S^{d-1}$ 
with cardinality $5^d$ such that 
\begin{equation*}
    \max_{\x\in\mathbb S^{d-1}}\min_{{\bf z}\in\mathcal N}\|\x-{\bf z}\|\le \frac12.
\end{equation*}
Then by Lemma \ref{lem:covering-approximation}
\begin{equation}\label{eqn:cor-orthonormal-preserving-pointwise}
    \max_{\x\in\mathbb S^{d-1}}\left|\|\mphi\U\x\|^2-1\right|\le 4\max_{\x\in\mathcal N}\left|\|\mphi\U\x\|^2-1\right|.
\end{equation}
By (\ref{eqn:in-cor-orthonormal-preserving2}), 
(\ref{eqn:cor-orthonormal-preserving-pointwise}) 
and union bound, 
\begin{equation*}
    \probP(\max_{\x\in\mathcal N}\left|\|\mphi\U\x\|^2-1\right|>\varepsilon)\le 2\cdot 5^d {\rm e}^{-\tilde c\varepsilon^2 n}.
\end{equation*}
The proof is completed once we note that $2\cdot 5^d\le e^{3d}$ for $d\ge1$. 

% ------------------------------------------------------

\section{Appendix: Randomly Sampled BOS and Partial Circulant Matrices}\label{apd:fast-matrices}
The purpose of this appendix is to prove Lemma \ref{lem:BOS-orthonormal-preserving} 
and Lemma \ref{lem:circulant-orthonormal-preserving}.
We will need Theorem \ref{thm:apd:rip-to-jl} and Theorem \ref{thm:apd:rip-of-bos} 
as stated in Section \ref{apd:jl-property} as well as the covering argument adapted there. 

\subsection{Randomly Sampled BOS}
\begin{proof}[Proof of Lemma \ref{lem:BOS-orthonormal-preserving}]
  By Theorem \ref{thm:apd:rip-of-bos}, 
  the restricted isometry constants of $\mphi$ satisfy
  \begin{equation}\label{eqn:bos-small-rip}
    \probP\left(\delta(s)\le\frac{\varepsilon}{4}\right)\ge1-\exp\left(-C^{-1}\varepsilon^2K^{-2}\frac ns\right),
  \end{equation}
  given
  \begin{equation}\label{eqn:large-n}
    n\ge CK^2\varepsilon^{-2}s\log^2s\log(K^2\varepsilon^{-2}s\log N)\log N.
  \end{equation}
  If $\delta(s)\le\varepsilon/4$ holds, 
  then by Theorem \ref{thm:apd:rip-to-jl} and a standard covering argument 
  (c.f. proof of Lemma \ref{lem:JL-orthonormal-preserving}), 
  $1-\varepsilon<s_{\min}(\mphi\U)\le s_{\max}(\mphi\U)<1+\varepsilon$ holds
  with probability at least $1-\mathrm e^{-cs+3d}$. 
  Thus by union bound, 
  \[1-\varepsilon<s_{\min}(\mphi\U)\le s_{\max}(\mphi\U)<1+\varepsilon\]
  holds with probability at least 
  \[
    1-\exp(-cs+3d)-\exp\left(-C^{-1}\varepsilon^2K^{-2}\frac ns\right)
  \]
  if \eqref{eqn:large-n} holds. 

  Set 
  \begin{equation*}
    s=\left\lceil\frac{3d+\sqrt{9d^2+4cC^{-1}\varepsilon^2K^{-2}n}}{2c}\right\rceil
  \end{equation*}
  Then for $n>C'\varepsilon^{-2}K^2$ the probability above is at least 
  \[
    1-2\exp\left(\frac34d-\frac14\sqrt{9d^2+C'^{-1}\varepsilon^2K^{-2}n}\right),
  \]
  as desired. 
  It remains to check that \eqref{eqn:large-n} holds. 
  Note that $s\ge C'^{-1}\varepsilon K^{-1}\sqrt n$. 
  For $n\ge C'\log N$ we have $K^2\varepsilon^{-2}\log N\le s^2$, 
  thus $\log(K^2\varepsilon^{-2}s\log N)\le 3\log s$. 
  It then suffices to show
  \begin{equation}\label{eqn:large-n-final-form}
    n\ge CK^2\varepsilon^{-2}s\log^3s\log N.
  \end{equation}
  But $s\le C'\max\{d, \varepsilon K^{-1}\sqrt n\}$, 
  which implies \eqref{eqn:large-n-final-form} 
  when $n\ge C'K^3\varepsilon^{-3}\max\{d\log^3d,\log^3N\}$.
\end{proof}

\subsection{Partial Circulant/Toeplitz Matrices}
\begin{proof}[Proof of Lemma \ref{lem:circulant-orthonormal-preserving}]
  The equations (3.3)-(3.7) in \cite{Vybiral2011Variant}, 
  together with (3.8) there replaced by Hanson-Wright inequality to control 
  the norm of $\|\vect\Sigma\V^*\vect a\|^2$,
  imply that 
  for any $\x\in\mathbb S^{N-1}$, we have $|\|\mphi\x\|^2-1|<\varepsilon$ 
  with probability at least 
  \[1-4N\mathrm e^{-\frac{t}4}-2\mathrm e^{-\frac{cn\varepsilon^2}{t}}\]
  for any $t>0$. 
  Take $t=\varepsilon\sqrt n$, the above probability is at least $1-4N\mathrm e^{-c\varepsilon\sqrt n}$.
  By a standard covering argument (c.f. Proof of Corollary \ref{cor:JL-implies-SRIP}), 
  $1-\varepsilon<s_{\min}(\mphi\U)\le s_{\max}(\mphi\U)<1+\varepsilon$ holds with probability at least 
  $1-4N\cdot 5^{d}\cdot\mathrm e^{-c\varepsilon\sqrt n}$, 
  which is greater than $1-\mathrm e^{-c\varepsilon\sqrt n/2}$ 
  when $n>C\varepsilon^{-2}(d+\log N)^2$.
\end{proof}

% ------------------------------------------------------

\section{Appendix: Random Matrices With Heavy-Tailed Distributions}\label{apd:heavy-tails}
The purpose of this appendix is 
to provide some material on heavy-tailed distributions, in particular, 
distribution with finite moments characterized by strong regularity condition 
and log-concave ensembles, 
and to provide a proof of Lemma \ref{lem:strong-regularity-orthonormal-preserving} 
and Lemma \ref{lem:log-concave-orthonormal-preserving}.

\subsection{Finite Moments}
Our proof depends on a theorem from \cite{Srivastava2013Covariance}, 
which asserts that 
\begin{theorem}[\cite{Srivastava2013Covariance}]\label{thm:s&v}
  Consider independent isotropic random vectors $\x_i$ valued in $\mathbb R^d$. 
  Assume that $\x_i$ satisfies the strong regularity assumption: 
  for some $C',\eta>1$, one has 
  \begin{equation*}
    \probP(\|\projP\x_i\|^2>t)\le C't^{-\eta},\quad\text{for $t>C'\operatorname{rank}\projP$}
  \end{equation*}
  for every orthogonal projection $\projP$ in $\mathbb R^d$. 
  Then there exists a polynomial function $\operatorname{poly}(\cdot)$ 
  whose coefficients depend only on $C'$ and $\eta$, 
  such that for any $\varepsilon\in(0,1)$ and for $n>\operatorname{poly}(\varepsilon^{-1})d$, 
  we have 
  \begin{equation*}
    \mathbb E\left\|\frac1n\sum_{i=1}^n\x_i\x_i^{\mathrm T}-\mathbf I\right\|\le\varepsilon.
  \end{equation*}
\end{theorem}
\begin{proof}[Proof of Lemma \ref{lem:strong-regularity-orthonormal-preserving}]
  Let $d=\lceil\alpha n\rceil$ where $\alpha\in(0,1)$ is to be determined later 
  (we shall choose some $\alpha$ that does not depend on $n$). 
  Fix a $d$-dimensional subspace of $\mathbb R^N$ and denote by $\U$ any of its orthonormal basis; 
  it follows that $\U$ is an $N\times d$ matrix. 
  We shall prove \eqref{eqn:orthonormal-perturbation-bound} for some $\delta$ and $\varepsilon$ 
  using Theorem \ref{thm:s&v} and partial strong regularity condition \eqref{eqn:strong-regularity}. 
  Denote the rows of $\mphi$ by $\frac{1}{\sqrt n}\x_1^{\mathrm T},\ldots,\frac{1}{\sqrt n}\x_n^{\mathrm T}$.
  Let $\y_i=\U^{\mathrm T}\x_i$. 
  Then $\y_i$'s are independent, centered and isotropic. 
  Furthermore, we have
  \begin{equation*}
    \|\U^{\mathrm T}\mphi^{\mathrm T}\mphi\U-\mathbf I\|=\left\|\frac{1}n\sum_{i=1}^n\y_i\y_i^{\mathrm T}-\mathbf I\right\|
  \end{equation*}

  Before applying Theorem \ref{thm:s&v}, 
  we need to show that $\y_i$ fulfills strong regularity condition. 
  For any orthogonal projection $\projP$ of rank $k$ in $\mathbb R^d$, 
  we consider the tail of $\|\projP\y_i\|=\|\projP\U^{\mathrm T}\x_i\|$. 
  First we note that there exists some $d\times k$ matrix $\V$ with orthonormal columns such that $\projP=\V\V^{\mathrm T}$. 
  Thus $\|\projP\U^{\mathrm T}\x_i\|=\|\U\projP\U^{\mathrm T}\x_i\|=\|(\U\V)(\U\V)^{\mathrm T}\x_i\|$. 
  But $\U\V$ is a matrix of rank$\le k$ with orthonormal columns, 
  since $\operatorname{rank}(\U\V)\le\operatorname{rank}\V$ 
  and $(\U\V)^{\mathrm T}(\U\V)=\V^{\mathrm T}\U^{\mathrm T}\U\V=\mathbf I$. 
  By \eqref{eqn:strong-regularity} we have
  \begin{equation*}
    \probP(\|\projP\y_i\|^2>t)\le C't^{-\eta},\quad\text{for $t>C'\operatorname{rank}(\U\V)$,}
  \end{equation*}
  hence for $t>C'k$. 
  This shows that $\y_i$ satisfies strong regularity condition. 
  The corollary follows from Theorem \ref{thm:s&v} and Chebyshev inequality.
\end{proof}

\subsection{Log-Concave Ensembles}
We will need the following well-known results\footnote{
    Sharper results are known in literature, e.g. \cite{Mendelson2014Singular}, 
    but this does not yield significant improvement in our case.
} on covariance estimation with log-concave ensembles.
\begin{theorem}[\cite{Adamczak2010Quantitative,Adamczak2011Sharp}]\label{thm:adamczak}
  Let $\x_1,\ldots,\x_n$ be independent centered isotropic random vectors 
  in $\mathbb R^d$ with log-concave distributions. 
  Then there exists some universal constants $c\in(0,1)$, $C>0$ such that
  \begin{equation*}
    \left\|\frac1n\sum_{i=1}^n\x_i\x_i^{\mathrm T}-\vect I\right\|\le C\sqrt{\frac dn}
  \end{equation*}
  with probability at least $1-2\exp(-c\sqrt d)$.
\end{theorem}

By definition, 
it is easy to check that the low-dimensional marginal of a log-concave distribution is log-concave.
\begin{proof}[Proof of Lemma \ref{lem:log-concave-orthonormal-preserving}]
  Let $\U$ be an orthonormal basis for a $d$-dimensional subspace of $\mathbb R^N$, 
  where $d=\lfloor C^{-2}\varepsilon^{2}n\rfloor$. 
  For $n>2C^2\varepsilon^{-2}d_2$ we have $d\ge 2d_2$, which suffices for our purpose.
  Set $\y_i=\U^{\mathrm T}\x_i$. 
  Argue as in the proof of Lemma \ref{lem:strong-regularity-orthonormal-preserving}, 
  we obtain that $\y_i$ is independent, centered, and isotropic, 
  and that
  \begin{equation}\label{eqn:log-concave-midstep}
    \|\U^{\mathrm T}\mphi^{\mathrm T}\mphi\U-\vect I\|=\left\|\frac1n\sum_{i=1}^n\y_i\y_i^{\mathrm T}-\vect I\right\|.
  \end{equation} 

  The distribution of $\y_i$ is a $d$-dimensional marginal of $\x_i$, 
  hence is log-concave.
  It follows from Theorem \ref{thm:adamczak} and \eqref{eqn:log-concave-midstep} that 
  \begin{equation*}
    1-\varepsilon<s_{\min}(\mphi\U)\le s_{\max}(\mphi\U)<1+\varepsilon
  \end{equation*}
  with probability at least $1-2\exp(-c\sqrt d)$. 
  When $n>10c^{-2}C^2\varepsilon^{-2}d_2$, we have $d>\frac12C^{-2}\varepsilon^2n+3c^{-2}$, 
  thus the probability above is at least $1-\exp(-c'\varepsilon\sqrt n)$, 
  as desired.
\end{proof}

% ------------------------------------------------------

\bibliographystyle{IEEEtran}
\bibliography{IEEEabrv,mybibfile}
\vfill\pagebreak
%\clearpage

\end{document}